\documentclass[11pt]{article}
\usepackage{amsmath,amsthm}
\usepackage{mathtools}
\mathtoolsset{showonlyrefs}

\usepackage{tabulary}
\usepackage{array,hhline}
\usepackage{paralist}

\usepackage{ifxetex,ifluatex}
\ifxetex
\usepackage{xltxtra}
\else
\ifluatex
\usepackage{luatextra}
\else
\usepackage{amssymb}
\usepackage[utf8]{inputenc}
\usepackage[T1]{fontenc}
\usepackage{mathpazo,tgcursor,tgtermes}
\usepackage[scale=.9]{tgheros}
\usepackage{textcomp}
\csname else\expandafter\endcsname
\romannumeral -`0
\fi
\fi
\iftrue\relax%
\defaultfontfeatures{Ligatures=TeX}
\usepackage[math-style=TeX, vargreek-shape=TeX]{unicode-math}
\AtBeginDocument{%
  \renewcommand{\varPhi}{\mitPhi}
  \let\setminus\smallsetminus
  
}
\fi

\usepackage[footnotesize,bf]{caption}

\usepackage[letterpaper,margin=1in]{geometry}
\usepackage{graphicx,authblk}
\usepackage[svgnames]{xcolor}
\usepackage{tikz}
\usetikzlibrary{backgrounds,fit}

\tikzset{vertex/.style={circle, fill, minimum size=10pt,
    inner sep=0pt},
  big vertex/.style={vertex, minimum size=20pt},
  container/.style={rectangle, rounded corners, draw=lightgray,
    thick}}

\usepackage[authoryear]{natbib}

\usepackage{xspace}
\usepackage{enumitem}
\setenumerate[1]{label=(\roman*)}
\usepackage{xparse}

\usepackage[bookmarksnumbered,bookmarksopen,unicode]{hyperref}

\DeclareMathOperator{\dom}{dom}
\DeclareMathOperator{\rank}{rank}
\DeclareMathOperator{\opt}{opt}
\newcommand{\nnegrk}{\rank_{+}} 
\newcommand{\psdrk}{\rank_{\oplus}} 

\newcommand{\LPrk}{\rank_{\textnormal{LP}}}
\newcommand{\SDPrk}{\rank_{\textnormal{SDP}}}
\newcommand{\R}{\mathbb{R}}
\newcommand{\N}{\mathbb{N}}

\DeclarePairedDelimiter{\abs}{\lvert}{\rvert}
\newcommand{\SDP}{\mathbb{S}_{+}}
\newcommand{\symM}{\ensuremath{\mathbb{S}}}

\DeclareMathOperator{\oprob}{\mathcal{P}}
\DeclareMathOperator{\fc}{fc_+}
\DeclareMathOperator{\fcp}{fc_{\oplus}}
\DeclareMathOperator{\tr}{Tr}

\providecommand{\cupdot}{\mathbin{\mathaccent\cdot\cup}}

\newcommand{\problem}[1]{\textnormal{\textsf{#1}}}
\newcommand{\MaxSAT}[1]{\problem{Max-\(#1\)-SAT}}
\newcommand{\MaxCONJSAT}[1]{\problem{Max-\(#1\)-CONJSAT}}
\newcommand{\MaxXOR}[1]{\problem{Max-\(#1\)-XOR}}

\newcommand{\MaxCUT}[1][]%
{\ensuremath{\problem{MaxCUT}\sb{#1}}\xspace}

\newcommand{\MaxMULTICUT}[1]{\problem{Max-MULTI-\(#1\)-CUT}}
\NewDocumentCommand{\VertexCover}{og}{%
  \ensuremath{%
    \IfValueTF{#2}
    {\problem{VertexCover(\(#2\))}}
    {\problem{VertexCover}}%
    \IfValueT{#1}{\sb{#1}}}%
  \xspace}
\NewDocumentCommand{\MaxIndep}{og}{%
  \ensuremath{%
    \IfValueTF{#2}
    {\problem{IndependentSet(\(#2\))}}
    {\problem{IndependentSet}}%
    \IfValueT{#1}{\sb{#1}}}%
  \xspace}

\newcommand{\MaxDicut}{\problem{Max-DICUT}\xspace}
\newcommand{\MaxCSP}[1]{\problem{Max-\(#1\)-CSP}}
\newcommand{\MinCNF}[1]{\problem{Min-\(#1\)-CNFDeletion}}
\newcommand{\MinUnCUT}{\problem{MinUnCUT}\xspace}
\newcommand{\GraphIsomorphism}{\problem{GraphIsomorphism}\xspace}

\DeclareMathOperator{\sat}{sat}
\DeclareMathOperator{\val}{val}

\newcommand{\face}[1]{\left\{#1\right\}}
\newcommand {\card}[1]{\left|#1\right|}
\newcommand {\binSet}{\{0,1\}}

\newcommand{\marknew}{(\textcolor{red}{new})}
\newcommand{\markoptimal}{(\textcolor{red}{optimal})}

\newenvironment{blockemph}{\begin{quote}\em}{\end{quote}}

\newcommand*{\set}[2]{\left\{#1\,\middle|\,#2\right\}}

\NewDocumentCommand{\conv}{mo}{\operatorname{conv}\left(#1%
    \IfValueT{#2}{\,\middle|\,#2}%
  \right)}
\NewDocumentCommand{\cone}{mo}{\operatorname{cone}\left(#1%
    \IfValueT{#2}{\,\middle|\,#2}%
  \right)}

\DeclarePairedDelimiterX\inp[2]{\langle}{\rangle}{#1,#2}

\NewDocumentCommand{\norm}{osm}{%
  \IfBooleanT{#2}{\left}\lVert
    #3%
    \IfBooleanT{#2}{\right}\rVert
  \IfValueT{#1}{\sb{#1}}}
\newcommand{\maxnorm}{\norm[\infty]}
\newcommand*{\size}[1]{\left|#1\right|}
\newtheorem{theorem}{Theorem}[section]
\newtheorem{lemma}[theorem]{Lemma}
\newtheorem{proposition}[theorem]{Proposition}
\newtheorem{corollary}[theorem]{Corollary}
\newtheorem{conjecture}[theorem]{Conjecture}
\theoremstyle{remark}
\newtheorem{remark}[theorem]{Remark}
\newtheorem{example}[theorem]{Example}
\theoremstyle{definition}
\newtheorem{definition}[theorem]{Definition}

\hypersetup{
  pdftitle={Affine reductions for LPs and SDPs},
  pdfauthor={Gábor Braun, Sebastian Pokutta, Daniel Zink},
  pdfkeywords={polyhedral approximation, extended formulation,
    stable sets},
  pdfsubject={MSC2000 Primary
    68Q17 
    90C05.
}}

\title{Affine reductions for LPs and SDPs}
\date{\today} 
\author[1]{Gábor Braun}
\affil[1]{ISyE, Georgia Institute of Technology,
  Atlanta, GA,
  USA.
  \textit{Email:}~gabor.braun@isye.gatech.edu}

\author[2]{Sebastian Pokutta}
\affil[2]{ISyE, Georgia Institute of Technology,
  Atlanta, GA,
  USA.
  \textit{Email:}~sebastian.pokutta@isye.gatech.edu}

\author[3]{Daniel Zink}
\affil[3]{ISyE, Georgia Institute of Technology,
  Atlanta, GA,
  USA.
  \textit{Email:}~daniel.zink@gatech.edu}

\begin{document}

\maketitle

\begin{abstract}
  We define a reduction mechanism for LP and SDP formulations that
  degrades approximation factors in a controlled fashion. Our
  reduction mechanism is a minor restriction of classical reductions
  establishing inapproximability in the context of PCP theorems.  As a
  consequence we establish strong linear programming inapproximability
  (for LPs with a polynomial number of constraints) for many
  problems. In particular we obtain a \(\frac{3}{2}-\varepsilon\)
  inapproximability for \VertexCover answering an open question in
  \cite{CLRS13} and we answer a weak version of our sparse graph
  conjecture posed in \cite{BFP2013} showing an inapproximability
  factor of \(\frac{1}{2}+\varepsilon\)
  for bounded degree \MaxIndep.  In the case of SDPs, we obtain
  inapproximability results for these problems relative to the
  SDP-inapproximability of \MaxCUT. Moreover,
  using our reduction framework we are able to reproduce various
  results for CSPs from \cite{CLRS13} via simple reductions from
  \MaxXOR{2}.
\end{abstract}

\section{Introduction}
\label{sec:introduction}

Linear Programming (LP) and Semidefinite Programming (SDP)
formulations are ubiquitous in combinatorial optimization problems  and
are often an important ingredient in the construction of approximation
algorithms. Extended formulations study the size of such LP or SDP
formulations. So far most strong lower bounds have been obtained by
ad-hoc analysis, which is stark contrast to the computational
complexity world, where we often resort to reduction mechanisms to
establish hardness or hardness of approximation. 

In this work we establish a strong reduction mechanism for approximate
LP and SDP formulations. Having a reduction mechanism in place has
three main appeals. 
\begin{enumerate}
\item The reduction mechanism allows for the propagation of hardness
results across optimization problems very similar in spirit to the
computational complexity approach. In particular, no specific
knowledge of how the hardness of the base problem has been established
is required. As such it turns establishing LP/SDP hardness into a
routine task for many problems allowing a much broader community to
benefit from results from extended formulations. 
\item Any future improvements to the strength of the lower bounds
  of the base problems immediately propagate to all problems that one
  can reduce to. In particular, it is likely that the lower bounds in
  \cite{CLRS13} and \cite{LRS14}, which form the base hard problems
  that we reduce from can be further improved. 
\item Having a reduction mechanism in place provides an ordering on the
  hardness of problems. As such it is a first step
  in identifying the LP/SDP analog to a complete problem. This is a
  quite appealing problem in its own right as it will likely have to
  reconcile the LP-hardness of matching and \(3\)-SAT, which is in
  contrast to the polynomial time solvability of matching in
  computational complexity. 
\end{enumerate}

Our reduction mechanism is compatible with
numerous known reductions in the literature that have been used to
analyze computational complexity, so that we can reuse these
reductions in the context of LPs and SDPs. As the approach is very similar in
all cases, i.e., to verify that the additional conditions of our
mechanism are met, we will only state a select few that are of
specific interest. We obtain new LP inapproximability results for problems
such as, e.g., \VertexCover, \MaxMULTICUT{k}, and bounded degree
\MaxIndep which are not 0/1 CSPs and hence not captured by the
approach in \cite{CLRS13}. Moreover we reproduce
previous results for CSPs in \cite{CLRS13} via direct reductions from
\MaxXOR{2}. This is interesting in view of \MaxXOR{2} being the actual
driver of complexity for most of these problems. In particular,
establishing a stronger lower bound for \MaxXOR{2}, e.g., of the form
\(2^{\Omega(n^\delta)}\) would imply improve the LP-hardness of
approximation of many other CSPs.

Our reduction mechanism is based on a more abstract view of extended
formulations that is motivated by the earlier approach in
\cite{CLRS13} to capture linear programming formulations independent
of the specific linear encoding in the context of CSPs, where the
feasible region is the whole 0/1 cube. We generalize this approach to
studying arbitrary combinatorial optimization problems and the
resulting model captures previous approaches studying approximations
via polyhedral pairs (see \cite{bfps2012,bfps2012jour}). In the same
spirit as \cite{CLRS13} we do not rely on lifting a specific
polyhedral representation but rather work directly with the
combinatorial optimization problem, i.e., we answer questions of the
form

\begin{blockemph}
  Given a combinatorial optimization problem, what is the smallest
  size of any of its LPs or SDPs?
\end{blockemph}

While this difference to the traditional extended formulation model is
more of a philosophical nature and the results are equivalent to those
obtained via the traditional extended formulations setup, on a
technical level this perspective significantly simplifies the
treatment of approximate LP/SDP formulations and it enables the
formulation of the reduction mechanism.

\subsection*{Related work}
\label{sec:related-work}

For a detailed account on various lower bounds for
extended formulations of combinatorial optimization
problems, see e.g., \cite{extform4,Rothvoss13}.
A reduction mechanism for \emph{exact} extended formulations
has been considered
in \cite{2013arXiv1302.2340A,VP2013},
where it already provided lower bounds on the exact
extension complexity of various polytopes, however both fall short to
capture approximations as they cannot incorporate the necessary affine
shifts. 
Our generalization of (encoding dependent) extended
formulations to encoding independent
\emph{formulation complexity} is a natural extension of \cite{CLRS13} and
\cite{BFP2013jour}. The former studied an encoding-independent model
for \emph{uniform formulations of CSPs}, whereas the latter used a restricted
LP version of the model presented here. 
Recently, there has been also
progress in relating general LP and SDP formulations with hierarchies
(\cite{CLRS13,lee2014power,LRS14}),
and we reuse some of the results as the
basis for later reductions. In particular \cite{LRS14} recently established
  super-polynomial lower bounds for the size of approximate SDP
  formulations capturing \MaxSAT{3} and other problems,
  however the SDP hardness of approximation of our base problems is
  still open. Our approach is also related to
\cite{KWW2013} as well as inapproximability reductions
in the context of PCPs and Lasserre hierarchies (see e.g.,
\cite{haastad1999clique,trevisan2000gadgets,haastad2001some,
  trevisan2004inapproximability,Schoen-k-CSP,tulsiani2009csp}).

\subsection*{Contribution}
\label{sec:contribution}

Our contribution can be broadly separated into the following three
parts, where the reduction mechanism is the main contribution. The
abstract view on extended formulations should be considered an enabler
for the reduction mechanism and the analysis of approximate LP/SDP
formulations. The use of the reduction mechanism is exemplified via
various new inapproximability results.

We stress that all results are \emph{independent} of P vs. NP and
pertain to solving combinatorial problems with LPs or SDPs.

\begin{enumerate}[leftmargin=1.5em]
\item \textbf{Factorization Theorem for Combinatorial Problems.}  The
  key element in the analysis of extended formulations is
  Yannakakis's celebrated Factorization Theorem (see
  \cite{Yannakakis91,Yannakakis88}) and its generalizations (see e.g.,
  \cite{GouveiaParriloThomas2011,bfps2012,bfps2012jour,CLRS13})
  equating the minimal size of an extended formulation with a property
  of a slack matrix, e.g., in the linear case the nonnegative rank.
  We provide an abstract, unified version
  (Theorem~\ref{thm:factorization}) of these factorization theorems
  for combinatorial optimization problems \emph{and their
    approximations} acting directly on the problem. From
  an optimal factorization one can explicitly reconstruct an
  optimal encoding as a linear program or semidefinite program.
  As a consequence we characterize
  the linear programming complexity as well as
  the semidefinite programming complexity of a combinatorial problem
  \emph{independent} of the linear representation of the problem. Our
  employed model is essentially equivalent to the extended formulations
  approach as we will see (see Section~\ref{sec:relat-appr-EF}) and its main purpose
  is to facilitate the definition (and use) of the 
reduction mechanism.
  
\item \textbf{Reduction mechanism.}  We provide a purely combinatorial
  and conceptually simple framework for reductions (similar to
  $L$-reductions) of optimization problems in the context of LPs and
  SDPs (and in fact \emph{any} other conic programming paradigm),
  where approximations are inherent without the need of any polyhedra
  (and polyhedral pairs) as compared to e.g.,
  \cite{bfps2012,bfps2012jour}.  Thereby we overcome many technical
  difficulties that prevented direct reductions for approximate LPs or SDPs
  in the past.  Many reductions in the context of PCP
  inapproximability are compatible with our mechanism and hence can be
  reused.  Contrasting the above, so far LP inapproximability results have been only obtained
  for very restricted classes of problems (see
  \cite{bfps2012,CLRS13,bfps2012jour}) requiring a case-by-case
  analysis.  

\item \textbf{LP inapproximability and conditional SDP
    inapproximability of specific problems.}  Our reduction mechanism
  opens up the possibility to reuse previous hardness results to
  establish inapproximability of problems. As a case in point, we
  establish the first LP inapproximability result for \VertexCover and
  \MaxIndep as well as reproduce the results in \cite{CLRS13}, all by
  simple and direct reductions from the LP-inapproximability of
  \MaxXOR{2} within a factor better than \(\frac{1}{2}\)
  established in \cite{CLRS13}. For SDP formulations our reductions
  establish relative inapproximability between the considered problems
  as strong SDP hardness of approximation for our base problems
  \MaxXOR{k} are unknown. Several of our inapproximability
  results are new and others reproduce the results in \cite{CLRS13}
  via direct reductions from a single hard problem.  The conditional SDP
  inapproximability factors are formulated under the assumption that the
  Goemans-Williams SDP for \MaxCUT is optimal (see
  Conjecture~\ref{conj:SDP-base}), which is compatible with the Unique
  Games Conjecture. Alternatively, these reductions can be also
  combined with the \(15/16+\varepsilon\)
  SDP-hardness of approximation for \MaxCUT that was recently
  established in \cite{braun2015strong} to obtain unconditional (but
  weaker) inapproximability factors. To obtain the respective factors
  it suffices to replace \(c_{GW}\) with \(15/16\) in the  arguments. 

In particular,
  we answer an open question regarding
  the inapproximability of \VertexCover (see
  \cite{CLRS13}) and we answer a weak version of our
sparse graph conjecture posed in \cite{BFP2013}. We obtain results as
provided in the following table:

  \begin{center}
    \begin{tabulary}{.9\textwidth}{ L | C | C | c | c }
      & \multicolumn{3}{ c| }{Inapproximability} & Approximability  \\
      Problem
      & (LP) & (SDP under \ref{conj:SDP-base}) & (PCP) & (LP)
  \\
  \hhline{=|=|=|=|=}
  \VertexCover \marknew & $\frac{3}{2} - \varepsilon$
  & \(1.12144 - \varepsilon\)
  & $1.361 - \varepsilon$
  & $2$ \\
  \MinCNF{2} \marknew & \(\omega(1)\) & \(\omega(1)\) &
  \(2.889 - \varepsilon\) &
  --- \\
  \MinUnCUT \marknew & \(\omega(1)\) & \(\omega(1)\) &
  \(C_{\text{uncut}} - \varepsilon\) &
  --- \\
  \MaxMULTICUT{k} \marknew &
  \(\frac{2 c(k) + 1}{2 c(k) + 2} + \varepsilon\) &
  \(\frac{c(k) + c_{GW}}{c(k) + 1}\) &
  \(1 - \frac{1}{34 k} + \varepsilon\) &
  \(\frac{1}{2(1-1/k)}\) \\
  bounded degree \MaxIndep \marknew & \(\frac{1}{2} + \varepsilon\)
  & \(c_{GW} + \varepsilon\)
  & \(O\left( \frac{\log^{4} \Delta}{\Delta} \right)\) &
---  \\
  \MaxSAT{2} & $\frac{3}{4} + \varepsilon$ \markoptimal &
  $\frac{1 + c_{GW}}{2} + \varepsilon \approx 0.93928 +
  \varepsilon$ &
  $\frac{21}{22}+\varepsilon$ &
  $\frac{3}{4}$ \\
  \MaxSAT{3} & $\frac{3}{4} +
  \varepsilon$ & $\frac{1 + c_{GW}}{2} + \varepsilon \approx 0.93928 +
  \varepsilon$ &
  $\frac{7}{8} + \varepsilon$ & $\frac{19}{27} $
  \\
  \MaxDicut &  $\frac{1}{2} + \varepsilon$ \markoptimal &
  \(c_{GW} + \varepsilon\) &
  $\frac{12}{13} + \varepsilon$ & \(\frac{1}{2}\)
  \\
  \MaxCONJSAT{2} & \(\frac{1}{2} + \varepsilon\) \markoptimal &
  \(c_{GW} + \varepsilon\) & \(\frac{9}{10} + \varepsilon\)&
  \(\frac{1}{2}\) \\
\end{tabulary}
\end{center}
Here \(c_{GW} \approx 0.87856\) is the approximation factor
of the algorithm for \MaxCUT
from \cite{GoemansWilliamson95},
and \(c(k)\) is a constant depending on \(k\) defined in
Section~\ref{sec:problemmax-multi-k}.
\end{enumerate}

Note that our conditional SDP inapproximability of \MaxSAT{2} and
\MaxSAT{3} within is
\((1 + c_{GW}) / 2 + \varepsilon \approx 0.93928 + \varepsilon\),
better than the current best unconditional inapproximability of
\(7/8\)
from \cite[Theorem~1.5]{LRS14}; see Section~\ref{sec:futh-hardn-appr}
for details. Moreover, there exists a linear program for \MaxIndep
achieving an approximation guarantee of \(2 \sqrt{n}\),
which is indeed better than the \(n^{1-\varepsilon}\)
hardness of approximation by \cite{haastad1999clique}, so that we
cannot expect to obtain a \(n^{1-\varepsilon}\) hardness of approximation ; see
\cite{VertexCover2015} for a detailed discussion.

As a nice (minor) byproduct, the presented framework augments the results in
\cite{extform4,bfps2012,braverman2012information,bfps2012jour,Rothvoss13}
to be independent of the chosen linear encoding (see
Section~\ref{sec:examples}), which allows us to reuse these results
directly in
reductions. Also, approximations of the slack matrix
provide approximate programs for the optimization problem, as shown in
Theorem~\ref{thm:round-LP}. Our approach also immediately carries over to
symmetric formulations (see \cite{SDPmatching2015} for SDPs)
and other conic programming
paradigms; the details are left to the interested reader.

Finally, we would also like to note that our model and reductions have
been already applied to obtain inapproximability results for
approximate \GraphIsomorphism in \cite{GraphIso2015}, exponential
lower bound for symmetric SDPs for \problem{Matching} (see
\cite[Theorem~3.1]{SDPmatching2015}),
and\((2 - \varepsilon)\)-inapproximability
for \VertexCover together with inapproximability of \MaxIndep within
any constant factor in \cite{VertexCover2015} improving our
Theorem~\ref{thm:VertexCover} by a significantly more involved
argument. An extension of this reduction mechanism is presented in
\cite{braun2015strong} relaxing some of our conditions, which enables the
study of fractional optimization problems (such as e.g., Sparsest Cut).

\section{Optimization problems}
\label{sec:programs}

We intend to study the required size of a linear
program or semidefinite program capturing a combinatorial optimization
problem with specified approximation guarantees.  In our context an
optimization problem is defined as follows.

\begin{definition}[Optimization problems]
  \label{def:max-problem}
  An \emph{optimization problem}
  \(\oprob = (\mathcal{S}, \mathcal{F}, \val)\)
  consists of a set \(\mathcal{S}\)
  of \emph{feasible solutions} and a set \(\mathcal{F}\)
  of \emph{instances},
  together with a real-valued objective function
  \(\val \colon \mathcal{S} \times \mathcal{F} \to \R\).
\end{definition}

A wide class of examples consist of constraint satisfaction problems
(CSPs):

\begin{definition}[Maximum Constraint Satisfaction Problem (CSP)]
  \label{def:maxCSP}
  A \emph{constraint family} \(\mathcal{C} = \{C_1,\dots, C_m\}\)
   on the boolean variables \(x_1,\dots,x_n\) 
  is a family of
  boolean functions \(C_i\) in \(x_1,\dots,x_n\).
  The \(C_{i}\) are \emph{constraints} or \emph{clauses}.
  The problem \(\mathcal{P}(\mathcal{C})\)
  corresponding to a constraint family \(\mathcal{C}\) has
   \begin{enumerate}
   \item \textbf{feasible solutions}
      all 0/1 assignments \(s\) to \(x_1,\dots,x_n\);
   \item \textbf{instances}
     all nonnegative weightings \(w_{1}, \dotsc, w_{m}\)
     of the constraints \(C_{1},\dotsc, C_{m}\)
   \item \textbf{objective function}
     the weighted sum of satisfied constraints:
     \(\sat_{w_{1},\dotsc, w_{m}}(s) = \sum_i w_i C_i(s)\).
   \end{enumerate}
   The goal is to maximize the weights of satisfied constraints,
   in particular CSPs are maximization problems.
  A \emph{maximum Constraint Satisfaction Problem}
  is an optimization problem
  \(\mathcal{P}(\mathcal{C})\)
  for some constraint family \(\mathcal{C}\).
  A \(k\)-CSP is a CSP where every constraint depends
  on at most \(k\) variables.
\end{definition}
For brevity, we shall simply use CSP for a maximum CSP,
when there is no danger of confusion with a minimum CSP.
In the following, we shall restrict to instances
with 0/1 weights, i.e., an instance is a subset
\(L \subseteq \mathcal{C}\) of constraints,
and the objective function computes
the number
\(\sat_{L}(s) = \sum_{C \in L} C(s)\) of constraints in \(L\)
satisfied by assignment \(s\).
Restriction to specific instances clearly
does not increase formulation complexity.

As a special case,
the \MaxXOR{k} problem restricts to constraints,
which are XORs of \(k\) literals.
Here we shall write the constraints in the equivalent
equation form
\(x_{i_{1}} \oplus \dotsb \oplus x_{i_{k}} = b\),
where \(\oplus\) denotes the addition modulo \(2\).

\begin{definition}[\MaxXOR{k}]
  For fixed \(k\) and \(n\),
  the problem \(\MaxXOR{k}\) is
  the CSP for variables \(x_{1}, \dotsc, x_{n}\)
  and the family \(\mathcal{C}\) of all constraints
  of the form
  \(x_{i_{1}} \oplus \dotsb \oplus x_{i_{k}} = b\)
  with \(1 \leq i_{1} < \dotsb < i_{k} \leq n\)
  and \(b \in \{0, 1\}\).
\end{definition}

An even stronger important restriction is \MaxCUT, a subproblem of
\(\MaxXOR{2}\) as we will see soon. The aim is to determine the maximum size of cuts
  for all graphs \(G\) with \(V(G)=[n]\).

\begin{definition}[\MaxCUT]
  \label{def:maxCut}
  The problem \MaxCUT
  has instances all simple graphs \(G\) with vertex set
  \(V(G) = [n]\),
  and feasible solutions all cuts on \([n]\),
  i.e., functions \(s \colon [n] \to \{0,1\}\).
  The objective function \(\val\) computes
  the number of edges \(\{i, j\}\) of \(G\)
  cut by the cut, i.e.,
  with \(s(i) \neq s(j)\).

  The problem \(\MaxCUT[\Delta]\) is the subproblem of \MaxCUT
  considering only graphs \(G\)
  with maximum degree at most \(\Delta\).
\end{definition}

We have \(\val_{G}(s) = \sat_{L(G)}(s)\),
for the constraint set
\(L(G) = \set{x_{i} \oplus x_{j} = 1}{\{i,j\} \in E(G)}\),
realizing \MaxCUT as a
subproblem of \MaxXOR{2} with the same feasible solutions.

We are interested
in approximately solving an optimization problem \(\oprob\)
by means of a linear program or
a semidefinite program.
Recall that a typical PCP inapproximability result
states that it is hard to decide between
\(\max \val_{i} \leq S(i)\) and \(\max \val_{i} \geq C(i)\)
for a class of instances \(i\)
and some easy-to compute functions \(S\) and \(C\) usually refereed to
as \emph{soundness} and \emph{completeness}.
Here and below \(\max \val_{i}\) denotes
the maximum value of the function \(\val_{i}\)
over the respective set of
feasible solutions.
We adopt the terminology to linear programs and semidefinite programs.
We start with the linear case.
\begin{definition}[LP formulation of an optimization problem]
  \label{def:LP-formulation} 
  Let
  \(\oprob = (\mathcal{S}, \mathcal{F}, \val)\)
  be an optimization problem
  with
  real-valued functions \(C, S\) on \(\mathcal{F}\),
  called \emph{completeness guarantee}
  and \emph{soundness guarantee}, respectively.
  If \(\oprob\) is a maximization problem,
  then let \(\mathcal{F}^{\mathcal{S}} \coloneqq
  \set{f \in \mathcal{F}}{\max \val_{f} \leq S(f)}\)
  denote the set of instances, for which the maximum is
  upper bounded by soundness guarantee \(S\).
  If \(\oprob\) is a minimization problem,
  then let \(\mathcal{F}^{\mathcal{S}} \coloneqq
  \set{f \in \mathcal{F}}{\min \val_{f} \geq S(f)}\)
  denote the set of instances, for which the minimum is
  lower bounded by soundness guarantee \(S\).

  A \emph{\((C, S)\)-approximate LP formulation} of \(\oprob\)
  is a linear program \(A x \leq b\)
  with \(x \in \R^{d}\)
together with the following \emph{realizations}:
  \begin{enumerate}
  \item \textbf{Feasible solutions} as vectors \(x^{s} \in \R^{d}\)
    for every \(s \in \mathcal{S}\) so that 
  \begin{align}
    \label{eq:LP-contain}
    A x^{s} &\leq b \qquad \text{for all } s \in \mathcal{S},   
  \end{align}
  i.e., the system \(Ax \leq b\) is a relaxation (superset) of
  \(\conv{x^s \mid s \in \mathcal{S}}\).
  \item \textbf{Instances} as affine functions
    \(w^{f} \colon \R^{d} \to \R\)
    for every \(f \in \mathcal{F}^{\mathcal{S}}\)
    such that
    \begin{align}
      \label{eq:LP-linear}
      w^{f}(x^{s}) & = \val_{f}(s)      \qquad \text{for all } s \in
      \mathcal{S}, 
    \end{align}
    i.e., we require that the linearization \(w^{f}\) of \(\val_{f}\)
    is exact on
all \(x^s\) with \(s \in \mathcal{S}\).
  \item \textbf{Achieving guarantee \(C\)}
  via requiring
  \begin{align}
    \label{eq:LP-approx}
    \max  \set{w^{f}(x)}{A x \leq b} &\leq C(f)
    \qquad \text{for all } f \in \mathcal{F}^{\mathcal{S}},
  \end{align}
for maximization problems (resp. \(\min
  \set{w^{f}(x)}{A x \leq b} \geq C(f)\) for minimization problems).
  \end{enumerate}
  The \emph{size} of the formulation is the number of inequalities
  in \(A x \leq b\).
  Finally, the \emph{LP formulation complexity} \(\fc(\oprob, C, S)\)
  of the problem \(\oprob\) is
  the minimal size of all its LP formulations.
\end{definition}

For all instances
\(f \in \mathcal{F}\) soundness and completeness should satisfy
\(C(f) \geq S(f)\) in the case of maximization problems and \(C(f) \leq S(f)\)
in the case of minimization problems in order to capture the notion of
relaxations and we assume this condition in the remainder of the paper.

\begin{remark}
  We use affine maps instead of linear maps to
  allow easy shifting of functions.
  At the cost of an extra dimension and an extra equation,
  affine functions can be realized as linear functions.
\end{remark}
\begin{remark}[Inequalities vs. Equations]
  Traditionally in extended formulations, one would
  separate the description into equations and inequalities and one
  would only count inequalities.
  In our framework,
  equations can be eliminated by restricting to the affine space
  defined by them, and parametrizing it as a vector space.
  However, note that restricting to linear functions,
  one might need an equation to represent affine functions
  by linear functions.
\end{remark}

For determining the exact maximum of a maximization problem,
one chooses \(C(f) = S(f) \coloneqq \max \val_{f}\).
To show inapproximability
within an approximation factor \(0 < \rho \leq  1\),
one chooses guarantees satisfying \(\rho C(f) \geq  S(f)\).
This choice is motivated to be comparable with factors of
approximation algorithms finding a feasible solution \(s\)
with \(\val_{f}(s) \geq \rho \max \val_{f}\).
For minimization problem, \(C(f) = S(f) \coloneqq \min \val_{f}\)
in the exact case,
and \(\rho C(f) \leq S(f)\)
for an approximation factor \(\rho \geq 1\)
provided \(\val_{f}\) is nonnegative.
This model of \emph{outer} approximation
streamlines the models in \cite{CLRS13,BFP2013jour,lee2014power},
and also
captures, simplifies, and generalizes
approximate extended formulations from \cite{bfps2012,bfps2012jour};
see Section~\ref{sec:relat-appr-EF} for a discussion.

We will now adjust Definition~\ref{def:LP-formulation}
to the semidefinite case.
For symmetric matrices, as customary, we use the
Frobenius product as scalar product, i.e., \(\inp{A}{B} = \tr[A B]\).
Recall that the psd-cone is self-dual under this scalar product.

\begin{definition}[SDP formulation of an optimization problem]
  \label{def:SDP-formulation} 
  Let
  \(\oprob = (\mathcal{S}, \mathcal{F}, \val)\)
  be a maximization problem
  with
  real-valued functions \(C, S\) on \(\mathcal{F}\) and let \(\mathcal{F}^{\mathcal{S}} \coloneqq
  \set{f \in \mathcal{F}}{\max \val_{f} \leq S(f)}\)
  as in Definition~\ref{def:LP-formulation}.

  A \((C,S)\)-approximate \emph{SDP formulation} of
  \(\oprob\)
  consists of a linear map \(\mathcal{A} \colon \symM^{d} \to \R^{k}\)
  and a vector \(b \in \R^{k}\)
  (defining a semidefinite program
  \(\set{X \in \SDP^{d}}{\mathcal{A}(X) = b}\)).
  Moreover, we require the following \emph{realizations} of the
  components of \(\oprob\):
  \begin{enumerate}
  \item \textbf{Feasible solutions} as vectors \(X^{s} \in \SDP^{d}\)
    for every \(s \in \mathcal{S}\) so that
  \begin{align}
    \label{eq:SDP-contain}
    \mathcal{A}(X^{s}) = b
  \end{align}
  i.e., the system \(\mathcal{A}(X)  =  b, X \in \SDP^d\)
  is a relaxation of \(\conv{X^s}[s \in \mathcal{S}]\).
  \item \textbf{Instances} as affine functions
    \(w^{f} \colon \symM^{d} \to \R\)
    for every \(f \in \mathcal{F}^{\mathcal{S}}\)
    with
    \begin{equation}
      \label{eq:SDP-linear}
      w^{f}(X^{s}) = \val_{f}(s)      \qquad \text{for all } s \in
      \mathcal{S}, 
    \end{equation}
    i.e., we require that the linearization \(w^{f}\) of \(\val_{f}\)
    is exact on
all \(X^{s}\) with \(s \in \mathcal{S}\).
  \item \textbf{Achieving guarantee \(C\)}
    via requiring
  \begin{align}
    \label{eq:SDP-approx}
    \max  \set{w^{f}(X)}{\mathcal{A}(X^{s})  =  b,
      \ X^s \in \SDP^d}
    \leq C(f)
    \qquad \text{for all } f \in \mathcal{F}, 
  \end{align}
  for maximization problems,
  and the analogous inequality for minimization problems.
  \end{enumerate}

  The \emph{size} of the formulation is the parameter \(d\).
  The \emph{SDP formulation complexity} \(\fcp(\oprob, C, S)\) of the
  problem \(\oprob\) is the minimal size of all its SDP formulations.
\end{definition}

\subsection{Relation to approximate extended formulations}
\label{sec:relat-appr-EF}

Traditionally in extended formulations, one would start from an
initial polyhedral representation of the problem and bound the size of
its smallest possible lift in higher-dimensional space. In the linear
case for example, the minimal number of required inequalities would
constitute the extension complexity of that polyhedral representation.
Our notion of \emph{formulation complexity} can be understood as the
minimum extension complexity over all possible polyhedral encodings of
the optimization problem. This independence of encoding addresses
previous concerns that the obtained lower bounds are polytope-specific
or encoding-specific and alternative linear encodings (i.e., different
initial polyhedron) of the same problem might admit smaller
formulations: we showed that this is not the case. More precisely, in view of
the results from above the standard notion of extension complexity and
formulation complexity are essentially equivalent, however the more
abstract perspective simplifies the handling of approximations and
reductions as we will see in Section~\ref{sec:reductions}.

The notion of \emph{LP formulation}, its size, and
LP formulation complexity are closely related to \emph{polyhedral pairs}
and linear encodings 
(see \cite{bfps2012,bfps2012jour}, and also \cite{Pashkovich12}).
In particular, given a \((C, S)\)-approximate LP formulation of
a maximization problem \(\oprob\)
with linear program \(A x \leq b\),
representations \(\set{x^{s}}{s \in \mathcal{S}}\)
of feasible solutions
and \(\set{w^{f}}{f \in \mathcal{F}^{\mathcal{S}}}\) of
instances,
one can define a polyhedral pair encoding \(\oprob\)
    \begin{equation}
      \label{eq:minimal-pair}
      \begin{aligned}
        P &\coloneqq \conv{x^{s}}[s \in \mathcal{S}], \\
        Q &\coloneqq \set{x \in \R^d}{\inp{w^{f}}{x}
          \leqslant C(f),\ \forall f \in \mathcal{F}^{\mathcal{S}}\}}.
      \end{aligned}
    \end{equation}
Then for \(K \coloneqq \set{x}{A x \leq b}\),
we have \(P \subseteq K \subseteq Q\).
Note that there is no need for the approximating polyhedron \(K\)
to reside in extended space,
as \(P\) and \(Q\) already live there.

Put differently, the LP formulation complexity of \(\oprob\)
is the minimum
size of an extended formulation over all possible linear
encodings of \(\oprob\).
The semidefinite case is similar, with the only difference
being that \(K\) is now a spectrahedron,
being represented by a semidefinite program
instead of a linear program.

\section{Factorization theorem and slack matrix}
\label{sec:fact-theor-slack}

We provide an algebraic characterization of formulation complexity
via the \emph{slack matrix of an optimization problem},
similar in spirit to factorization theorems
for extended formulations
(see e.g.,
\cite{Yannakakis91,Yannakakis88,GouveiaParriloThomas2011,bfps2012,bfps2012jour}),
with a fundamental difference pioneered in \cite{CLRS13}
that there is no linear system to start from.
The linear or semidefinite program is constructed from scratch
using a matrix factorization.
This also extends \cite{BFP2013jour},
by allowing affine functions,
and using a modification of nonnegative rank,
to show that formulation complexity depends only on the slack matrix.

\begin{definition}[Slack matrix of \(\oprob\)]
  Let \(\oprob=(\mathcal{S}, \mathcal{F}, \val)\)
  be an optimization problem with guarantees \(C, S\).
  The \emph{\((C, S)\)-approximate slack matrix} of \(\oprob\)
is the nonnegative
  \(\mathcal{F}^{\mathcal{S}} \times \mathcal{S}\) matrix \(M\),
  with entries
  \begin{equation*}
    M(f, s) \coloneqq
    \begin{cases*}
      C(f) - \val_{f}(s) & if \(\oprob\) is a maximization problem, \\
      \val_{f}(s) - C(f) & if \(\oprob\) is a minimization problem.
    \end{cases*}
  \end{equation*}
\end{definition}

We introduce the \emph{LP factorization} of a nonnegative
matrix, which for slack matrices captures the LP formulation
complexity of the underlying problem.

\begin{definition}[LP factorization of a matrix]
  A \emph{size-\(r\) LP factorization} of \(M \in \R^{m \times  n}_+\)
  is a factorization
  \(M = TU + \mu \mathbb{1}\) where \(T \in \R_+^{m \times r}\),
  \(U \in \R_+^{r \times n} \) and \(\mu \in \R_{+}^{m \times 1}\).
  Here \(\mathbb{1}\) is the \(1 \times n\) matrix with all entries
  being \(1\).
  The \emph{LP rank \(\LPrk{M}\) of \(M\)} is the minimum
  \(r\) such that there exists a size-\(r\) LP factorization of
\(M\).
\end{definition}

A size-\(r\) LP factorization is
equivalent to a decomposition
\(M = \sum_{i \in [r]} u_iv_i^\intercal + \mu \mathbb{1}\) for
some (column) vectors \(u_i \in \R_+^{m}\), \(v_i \in \R_+^{n}\) with
\(i \in [r]\) and a column vector \(\mu \in \R_{+}^{m}\).
It is a slight modification of
a nonnegative matrix factorization, disregarding
simultaneous shift of all columns by the same vector,
i.e., allowing an additional
term \(\mu \cdot \mathbb{1}\) \emph{not} contributing to the
size, so clearly, \(\LPrk \leq \nnegrk M \leq \LPrk M + 1\).

One similarly defines SDP factorizations of nonnegative matrices.
\begin{definition}
  A \emph{size-\(r\) SDP factorization}
  of \(M \in \R^{m \times n}_+\)
  is a factorization is a collection of matrices
  \(T_1, \dots, T_m \in \SDP^r\)
  and \(U_1, \dots, U_n \in \SDP^r\)
  together with \(\mu \in \R_{+}^{m \times 1}\)
  so that \(M_{ij} = \tr[T_i U_j] + \mu(i)\).
  The \emph{SDP rank \(\psdrk{M}\) of \(M\)}
  is the minimum \(r\)
  such that there exists a size-\(r\) SDP factorization of \(M\).
\end{definition}

For the next theorem,
we need the folklore formulation of
linear duality using affine functions,
see e.g., \cite[Corollary~7.1h]{schrijver86._LP}.

\begin{lemma}[Affine form of Farkas's Lemma]
  \label{lem:Farkas}
  Let \(P \coloneqq \set{x}{A_{j} x \leq b_{j}, j \in [r]}\)
  be a non-empty polyhedron.
  An affine function \(\varPhi\) is nonnegative on \(P\)
  if and only if
  there are nonnegative multipliers \(\lambda_{j}\), \(\lambda_0\) with
  \begin{equation*}
    \varPhi(x) \equiv \lambda_{0}
    + \sum_{j \in [r]} \lambda_{j} (b_{j} - A_{j}x)
    .
  \end{equation*}
\end{lemma}

We are ready for the factorization theorem for optimization
problems.

\begin{theorem}[Factorization theorem for formulation complexity]
  \label{thm:factorization}
  Consider an optimization problem \(\oprob = (\mathcal{S},
  \mathcal{F}, \val)\) with \((C, S)\)-approximate slack matrix \(M\).
  Then we have
  \begin{equation*}
    \fc(\oprob, C, S) = \LPrk M, \qquad \text{and} \qquad
    \fcp(\oprob, C, S)= \SDPrk M.
  \end{equation*}
  for linear formulations and semidefinite formulations, respectively.
\end{theorem}

\begin{remark}
  The factorization theorem for polyhedral pairs (see
  \cite{bfps2012,Pashkovich12,bfps2012jour}) states that the
  nonnegative rank and extension complexity might differ by \(1\),
  which was slightly elusive. Theorem~\ref{thm:factorization}
  clarifies, that this is the difference between the LP rank and
  nonnegative rank, i.e., formulation complexity is a
  \emph{property of
  the slack matrix}.  Note also that for the slack matrix of a
  polytope, every row contains a \(0\)
  entry, and hence the \(\mu \mathbb{1}\)
  term in any LP factorization must be \(0\).
  Therefore the nonnegative rank and LP rank coincide for polytopes.
  Similar remarks apply to SDP factorizations.
\end{remark}

\begin{proof}[Proof of Theorem~\ref{thm:factorization}—the
linear case]
We will confine ourselves to the case of \(\oprob\) being a
maximization problem.
For minimization problems, the proof is analogous.

To prove \(\LPrk M \leq \fc(\oprob, C, S)\),
let \(A x \leq b\) be an arbitrary \((C, S)\)-approximate,
size-\(r\) LP formulation of \(\oprob\),
with realizations \(\set{w^{f}}{f \in \mathcal{F}^{\mathcal{S}}}\)
of instances
and \(\set{x^{s}}{s \in \mathcal{S}}\)
of feasible solutions.
We shall construct a size-\(r\)
nonnegative factorization of \(M\).
As \(\max_{x : Ax \leq b} w^{f}(x) \leq C(f)\)
by Condition~\eqref{eq:LP-approx},
via the affine form of Farkas's lemma,
Lemma~\ref{lem:Farkas}
we have
\[
C(f) - w^{f}(x)
= \sum_{j=1}^r T(f,j) \left( b_j - \inp{A_j}{x} \right) + \mu(f)
\]
for some nonnegative multipliers $T(f,j), \mu(f) \in \R_+$
with \(1 \leq j
\leq r\).
By taking $x = x^s$,
we obtain
\begin{align}
  \label{eq:1}
  M(f,s) &= \sum_{j=1}^r T(f,j) U(j,s) + \mu(f),
  &
  \text{with} \quad
  U(j,s) &\coloneqq
    b_j - \inp{A_j}{x^s} & \text{for } j > 0.
\end{align}
i.e., $M = TU + \mu \mathbb{1}$.
By construction, $T$ and \(\mu\) are nonnegative.
By Condition~\eqref{eq:LP-contain} we also obtain that $U$ is nonnegative.
Therefore $M = TU + \mu \mathbb{1}$
is a size-$r$ LP factorization of $M$.

For the converse, i.e., \(\LPrk \geq \fc(\oprob, C, S)\),
let \(M = TU + \mu \mathbb{1}\) be a size-\(r\) LP factorization.
We shall construct an LP formulation of size \(r\).
Let \(T_{f}\) denote the \(f\)-row of \(T\)
for \(f \in \mathcal{F}^{\mathcal{S}}\),
and \(U_{s}\) denote
the \(s\)-column of \(U\) for \(s \in \mathcal{S}\).
We claim
that the linear system
\(x \geq 0\)
with representations
\begin{align*}
  w^{f}(x) &\coloneqq C(f) - \mu(f) - T_{f} x
  \qquad \forall f \in \mathcal{F}^{\mathcal{S}}
  &&\text{and}&
  x^{s} &\coloneqq U_{s} \qquad \forall s \in \mathcal{S}
\end{align*}
satisfies the requirements of
Definition~\ref{def:LP-formulation}.
Condition~\eqref{eq:LP-linear} is implied by
the factorization \(M = T U + \mu \mathbb{1}\):
\begin{equation*}
  w^{f}(x^{s}) = C(f) - \mu(f) - T_{f} U_{s} = C(f) - M(f, s) =
 \val_{f}(s).
\end{equation*}
Moreover, \(x^{s} \geq 0\),
because \(U\) is
nonnegative, so that Condition~\eqref{eq:LP-contain} is fulfilled.
Finally, Condition~\eqref{eq:LP-approx} also follows readily:
\begin{equation*}
  \max \set{w^{f}(x)}{x \geq 0}
  = \max \set{C(f) - \mu(f) - T_{f} x}{x \geq 0} = C(f) - \mu(f)
  \leq C(f),
\end{equation*}
as the nonnegativity of \(T\) implies \(T_{f} x \geq 0\); equality
holds e.g., for \(x = 0\).
Recall also that \(\mu(f) \geq 0\).
Thus we have constructed an LP formulation with \(r\) inequalities,
as claimed.
\end{proof}

\begin{remark}
  It is counter-intuitive that \(0\) is always a maximizer,
  and, actually, it is an artifact of the construction.
  At a conceptual level,
  the polyhedron \(A x \leq b\) containing
  \(\conv{x^s \mid s \in \mathcal{S}}\)
  is represented as the intersection of the nonnegative cone
  with an affine subspace in the slack space.
  The affine functions \(w^{f}\) are extended to attain
  their optimum value on this intersection in the nonnegative cone,
  and thus also at \(0\), the apex of the cone.
  In particular, intersecting with the affine subspace is no longer
  needed.
  See Figure~\ref{fig:optimal-LP} for an illustration.
\end{remark}

\begin{remark}[Solution structure] 
Observe that the obtained LP
  formulation via the LP-factorization of the slack matrix also
  separates solutions \(s \in S\) into two disjoint classes \(S =
  S_{o} \cupdot S_n\), where \(S_o\) contains those solutions that
  potentially can be optimal for some function, i.e., it is the set of
  coordinate-wise minimal points in \(S\). The set \(S_n\) is
   the set of solution that are \emph{never} optimal for any \(f \in
  F\) as they are coordinate-wise dominated by at least one point in
  \(S_o\), i.e., for any \(s_n \in S_n\)
  there exists \(s_o \in S_o\) with \(s_o \leq s_n\) coordinate-wise. 

This might have applications in the context of inverse optimization
(see e.g., \cite{ahuja2001inverse}), where we would like to decide whether for a
given solution \(s \in \mathcal S\), there exists an \(f \in \mathcal
F\), that is maximized at \(s\). This can now be read off the factorization.
\end{remark}

\begin{figure}
  \centering
  \begin{tikzpicture}[scale=.9, >=latex,
    generatrix/.style={dashed}]
    \draw[->] (-1,0) -- (10, 0);
    \draw[->] (0, -1) -- (0, 10);
    \draw[->] (-.4, -.5) -- (2, 2.5);
    \draw[fill=gray]
    (1, 9/1) edge[generatrix] (0,0) --
    (2, 9/2) edge[generatrix] (0,0) --
    (3, 9/3) edge[generatrix] (0,0) --
    (5, 9/5) edge[generatrix] (0,0) --
    (9, 9/9) edge[generatrix] (0,0)
    --(7,7) node [right=2mm] (solution) {\(x^{s}\)} -- (5,9) --cycle;
    \node (conv) at (4.5,6) {\(\conv{x^{s}}[s \in \mathcal{S}]\)};
    \filldraw (9.5,4) circle[radius=4pt]
    edge (1, 9/1)
    edge (2, 9/2)
    edge (3, 9/3)
    edge (5, 9/5)
    edge (9, 9/9)
    edge (7, 7)
    edge (5, 9)
    node[right=2mm] {\(S_{n}\)};
    \draw[very thick] (-.5, 8) -- (4,-.5);
    \draw[->]
    (1.75, 3.75) -- +(-.85, -.45) node[below=1mm]{\(w^{f}\)};
  \end{tikzpicture}
  \caption{Linear program from an LP factorization. The LP is the
    positive orthant \(x \geq 0\). The point \(0\) is a maximizer for
    all linearizations \(w^f = C(f) - \mu(f) - T_f x\)
    of objective functions \(\val_{f}\) for all instances \(f\).
    The normals of all objective functions point in
    nonpositive direction as \(T_f \geq 0\).}
  \label{fig:optimal-LP}
\end{figure}
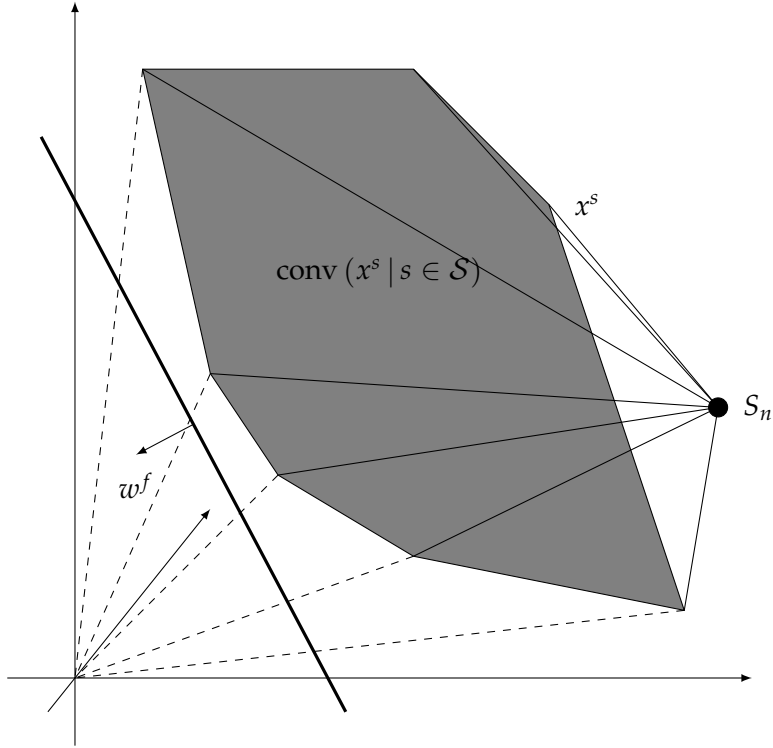

\begin{proof}[Proof of Theorem~\ref{thm:factorization}—the
  semidefinite case]
As before we confine ourselves to the case of \(\oprob\) being a
maximization problem.
The proof is analogous to the linear case,
but for the sake of completeness, we provide a full proof.

To prove \(\SDPrk M \leq \fcp(\oprob, C, S)\),
let \(\mathcal{A}(X) = b, X \in \SDP^{r}\) be an arbitrary size-\(r\)
SDP formulation of \(\oprob\),
with realizations \(\set{w^{f}}{f \in \mathcal{F}^{\mathcal{S}}}\)
of instances
and \(\set{X^{s}}{s \in \mathcal{S}}\)
of feasible solutions.
To apply strong duality,
we may assume that the convex set
\(\set{ X \in \SDP^{r}}{\mathcal{A}(X) = b}\)
has an interior point because otherwise
it would be contained in a proper face of \(\SDP^{r}\),
which is an SDP cone of smaller size.
We shall construct a size-\(r\)
SDP factorization of \(M\).
As \(\max_{X \in \SDP^{r} \colon \mathcal{A}(X) = b }
w^{f}(X) \leq C(f)\)
by Condition~\eqref{eq:SDP-approx},
via the affine form of strong duality,
we have
\[
C(f) - w^{f}(X)
 = \inp{T_{f}}{X} + \inp{y^{f}}{b - \mathcal{A}(X)} + \lambda_{f}
\]
for all \(f \in \mathcal F\)
with some $T_{f} \in \SDP^{r}$, \(y^{f} \in \R^k\) and
\(\lambda_{f} \in \R_{+}\).
By substituting \(X^{s}\) into \(X\),
we obtain
\begin{equation}
  \label{eq:SDP1}
  M(f,s)
  = C(f) - w^{f}(X^{s})
  = \inp{T_{f}}{X^{s}} + \inp{y^{f}}{b - \mathcal{A}(X^{s})}
  + \lambda_{f}
  = \inp{T_{f}}{X^{s}} + \lambda_{f},
\end{equation}
which is an SDP factorization of size \(r\).

For the converse,
i.e., \(\fcp(\oprob, C, S) \leq \SDPrk M\),
let \(M(f, s) = \inp{T_{f}}{U_{s}} + \mu(f)\) be
a size-\(r\) SDP factorization.
We shall construct an SDP formulation of size \(r\).
We claim that the SDP formulation: 
\begin{equation}
  \label{eq:SDP-factor-SDP}
  X \in \SDP^{r}
\end{equation}
with representations
\begin{align*}
  w^{f}(X) &\coloneqq
  C(f) - \mu(f) - \inp{T_{f}}{X}
  \qquad \forall f \in \mathcal{F}^{\mathcal{S}}
  &&\text{and}&
  X^{s} &\coloneqq U_{s}
  \qquad \forall s \in \mathcal{S}
\end{align*}
satisfies the requirements of
Definition~\ref{def:SDP-formulation}.
Condition~\eqref{eq:SDP-linear} follows by:
\begin{equation*}
  w^{f}(X^{s}) =  C(f) - \mu(f) - \inp{T_{f}}{U_{s}}
  = C(f) - M(f, s) = \val_{f}(s).
\end{equation*}
Moreover, the \(X^{s} = U_{s}\) are psd,
hence clearly satisfy the system
\eqref{eq:SDP-factor-SDP},
so that Condition~\eqref{eq:SDP-contain} is fulfilled.
Finally, Condition~\eqref{eq:SDP-approx} also follows readily.
\begin{equation*}
  \max \set{w^{f}(X)}{X \in \SDP^{r}}
  =
  \max \set{C(f) - \mu(f) - \inp{T_{f}}{X}}
  {X \in \SDP^{r}}
  = C(f) - \mu(f) \leq C(f),
\end{equation*}
as \(T_f\) and \(X\) being psd implies \(\inp{T_{f}}{X} \geq 0\);
equality holds e.g., for \(X = 0\).
Thus we have constructed an SDP formulation of size \(r\)
as claimed.
\end{proof}

\subsection{Examples}
\label{sec:examples}

In the context of optimization problems we typically differentiate two
types of formulations.  The \emph{uniform model} asks for a
formulation \emph{for a whole family of instances}.
Our Examples~\ref{ex:matching}, \ref{ex:Hamiltonian} and
\ref{ex:STAB-unif} are all uniform models.  The \emph{non-uniform
  model} asks for a formulation for the weighted version of
\emph{a specific problem},
where the instances differ only in the weighting.
Lower bounds or
inapproximability factors for non-uniform models are usually stronger
statements, as in the non-uniform case the formulation potentially
could adapt to the instance resulting in potentially smaller
formulations; see \cite{VertexCover2015} for such an example in the context of stable
sets. We refer the reader to \cite{CLRS13,BFP2013jour} for an
in-depth discussion.

The difference between uniform and non-uniform sometimes depends on
the point of view.
For graph problems, often the non-uniform model
for a graph $G$ induces a uniform model for the family of all of its
(induced) subgraphs by choosing 0/1 weights (see e.g.,
Definition~\ref{def:vertex-cover}).
In other words, the non-uniform model is actually a uniform model
for the class of all subinstances of \(G\).
Examples~\ref{ex:matching} and \ref{ex:Hamiltonian}
with \(0/1\) weights demonstrate this: they are uniform models  for all
subgraphs \(G \subseteq K_{n}\), but can also be viewed
as non-uniform models for \(K_{n}\).

Note that these models can also be used for studying
average case complexity.
For example, one might consider the complexity of the problem
for a randomly selected large class of instances using the uniform
model, or one might consider the non-uniform model for a randomly
selected instance.
For the maximum stable set problem, both random versions were examined
in \cite{BFP2013jour}.

\subsubsection{The matching problem revisited}
\label{sec:tsp-problem-matching}

The lower bounds in \cite{extform4,Rothvoss13} are concerned with
specific polytopes, namely the TSP polytope as well as the
matching polytope.
We obtain, as a slight generalization,
the same lower bounds for the Hamiltonian cycle problem (which is
captured by the
TSP problem with appropriate weights) as
well as the matching \emph{problem}, independent of choice of the
specific polytope that represents the encoding. 
Here we present only the matching problem.
The lower bound for the Hamiltonian cycle
problem will be obtained in Section~\ref{sec:face-reduction} via a
reduction. 

\begin{example}[Maximum matching problem]
  \label{ex:matching}
  The maximum matching problem \(\oprob_{\text{Match}}\) asks for the maximum size of matchings
  in a given graph.  While it can be solved in polynomial time, the
  matching polytope  has
  exponential extension complexity as shown in
  \cite{Rothvoss13}. Using the framework from above, we immediately
  obtain that the matching problem has high LP formulation complexity,
  reusing the lower bound on the nonnegative rank of the slack matrix
  of the matching polytope. 

  For the natural formulation of the problem in our framework,
  let \(n\) be fixed, and
  let the set \(\mathcal{S}\) of feasible solutions
  consist of all perfect matchings \(M\)
  of the complete graph on \(2n\) vertices,
  and the instances be all (simple) graphs \(G\)
  on \([2n]\).
  The value \(\val_{G}(M)\) for a graph \(G\) and
  a perfect matching \(M\)
  is defined to be
  \begin{equation*}
    \val_{G}(M) \coloneqq \size{M \cap E(G)}
  \end{equation*}
  the number of edges shared by \(M\) and \(G\),
  i.e., the size of the matching \(M \cap E(G)\) of \(G\).
  Clearly, all maximum matchings of \(G\) can be obtained this way
  (via extension to any perfect matching on \([2n]\)).
  Thus \(\max \val_{G}\)
  is the matching number \(\nu(G)\) of \(G\),
  the size of the maximum matchings of \(G\).

  Inspired by the description of the facets of the matching polytope
  in \cite{Edmonds65}, we only consider complete subgraphs
  on odd-sized subsets \(U\). For such a complete graph \(K_U\) on the
  odd-sized set \(U\), we have
  \(\max \val_{K_U} = \frac{\card{U}-1}{2}\).
  Let \(\delta(U)\) denote the set of all edges
  between \(U\) and its complement \([2n] \setminus U\).
  We have the identity \(\size{U} = 2 \size{M \cap E(K_{U})} + \size{M
    \cap \delta(U)}\) and 
thus obtain the slack matrix 
  for the exact problem
  (i.e., \(C(G) = \max \val_{G}\))
  \begin{equation*}
    S(K_{U}, M) \coloneqq C(K_{U})  - \val_{K_{U}}
    =  \frac{\size{U} - 1}{2} - \size{M \cap E(K_{U})}
    = \frac{\size{M \cap \delta(U)} - 1}{2}.
  \end{equation*}
  This submatrix has nonnegative rank \(2^{\Omega(n)}\) by
  \cite{Rothvoss13} and hence the LP formulation complexity of the
  maximum matching problem is \(2^{\Omega(n)}\), i.e.,
  \(\fc(\oprob_{\text{Match}}) = 2^{\Omega(n)}\). 

The result can be extended to the approximate case
with an approximation factor \((1 + \varepsilon / n)^{-1}\),
by invoking the lower bound for the resulting slack matrix from
\cite{BP2014matching}, showing that the maximum matching problem does
not admit any fully-polynomial size relaxation scheme.
Note that this is not unexpected as for the maximum matching problem
an approximation factor of about $1 - \frac{\varepsilon}{n}$
corresponds to an error
of less than one edge in the unweighted case for small $\varepsilon$,
so that the decision problem could be decided via the
approximation. This is a behavior similar to FPTAS and strong
NP-hardness for combinatorial optimization problems that are mutually
exclusive (under standard assumptions).
\end{example}

\subsubsection{Independent set problem}
\label{sec:stable-set-problem}

We provide an example for maximum independent sets in a uniform model; see
\cite{BFP2013jour} for more details as well as an average case analysis. 
Here there is no bound on the maximum degree of graphs,
unlike in Theorem~\ref{thm:VertexCover}.

\begin{example}[Maximum independent set problem (uniform model)]
  \label{ex:STAB-unif}
Let us consider the maximum
independent set problem \(\oprob\)
    over some family \(\mathcal{G}\) of graphs \(G\)
    where \(V(G) \subseteq [n]\) with aim to estimate the maximum size
    \(\alpha(G)\) of
    independent sets in each \(G \in \mathcal{G}\).

    A natural choice is to let the
    \emph{feasible solutions} be all subsets \(S\)
    of \([n]\), and the \emph{instances} be all
    \(G \in \mathcal{G}\).
    The objective function is
    \begin{equation*}
      \val_{G}(S) \coloneqq \size{V(G) \cap S} - \size{E(G(S))}.
    \end{equation*}
    Here \(\val_{G}(S)\) can be easily seen to lower bound
    the size of an independent set,
    obtained from \(S\) by removing vertices not in \(G\),
    and also removing one end point of every edge
    with both end points in \(S\).
    Clearly, \(\val_{G}(S) = \size{S}\) for independent sets \(S\) of
    \(G\), i.e., in this case our choice is exact.
    Thus \(\alpha(G) = \max_{S \subseteq [n]} \val_{G}(S)\).

Let us consider the special case when \(\mathcal{G}\)
is the set of all simple graphs with \(V(G) \subseteq [n]\).
We shall use guarantees
\(S(G) \coloneqq \max \val_{G} = \alpha(G)\)
and \(C(G) \coloneqq \rho^{-1} \val_{G}\)
for an approximation factor \(0 < \rho \leq 1\).
Restricting to complete graphs \(K_{U}\) with \(U \subseteq [n]\),
the obtained slack matrix is
a \((\rho^{-1} - 1)\)-shift of the (partial) unique disjointness
matrix, hence for approximations within a factor of \(\rho\), we
obtain the lower bound \(\fc(\oprob, C, \max) \geq 2^{\frac{n \rho}{8}}\)
with \cite{braverman2012information,BP2013}.   

See \cite{BFP2013jour} for other choices of \(\mathcal{G}\), such as
e.g., randomly choosing the graphs.
\end{example}

\subsubsection{$k$-juntas via LPs}
\label{sec:k-juntas-via}

It is well-known that the level-\(k\) Sherali–Adams hierarchy
captures all nonnegative \(k\)-juntas,
i.e., functions \(f \colon \binSet^n \rightarrow
\R_+\) that depend only on \(k\) coordinates of the input (see e.g.,
\cite{CLRS13}) and it can be written as a linear program using
\(O(n^k)\) inequalities. We will now show that this is essentially
optimal for \(k\) small. 

\begin{example}[\(k\)-juntas] We consider the problem of maximizing nonnegative
  \(k\)-juntas over the \(n\)-dimensional hypercube.
  Let the set of instances \(\mathcal{F}\)
  be the family of all nonnegative \(k\)-juntas and
  let the set of feasible solutions be \(\mathcal{S} = \binSet^{n}\),
  with \(\val_{f}(s) \coloneqq \val_{f}(s)\).
  We put \(C(f) = S(f) = \max_{s \in \mathcal S} \val_{f}(s)\).

As we are interested in a lower bound we will confine ourselves to a
specific subfamily of functions \(\mathcal F' \coloneqq \set{f_a}{a
  \in \binSet^n, \card{a} = k} \subseteq \mathcal F\) with 
\[f_a(b) \coloneqq a^\intercal b - 2 \binom{a^\intercal b}{2},\]
and hence \(C(f_{a}) = 1\).
Clearly \(\card{\mathcal F'} = \binom{n}{k}\), so that the nonnegative
rank of the slack matrix 
\[S_{a,b} \coloneqq C(f_{a}) - f_a(b) = 1 - a^\intercal b + 2
\binom{a^\intercal b}{2}  = (1 - a^\intercal b)^2,\]
with \(a, b \in \binSet^n\) and \(\card{a} = k\) is at most
\(\binom{n}{k}\).

Now for each \(f_a \in \mathcal F'\) we have that \(C(f_{a}) - f_a(b) =
(1 - a^\intercal b)^2 = 1\) if \(a \cap b = \emptyset\) and there are
\(2^{n-k}\) such choices for \(b\) for a given \(a\). Thus the matrix \(S\) has 
\(\binom{n}{k} 2^{n-k}\) entries \(1\) arising from disjoint pairs
\(a,b\). However in \cite{2013arXiv1307.3543K} it was shown that any
nonnegative rank-1 matrix can cover at most \(2^n\) of such
pairs. Thus the nonnegative rank of \(S\) is at least
\[\frac{\binom{n}{k}2^{n-k}}{2^n} = \frac{\binom{n}{k}}{2^k}.\]
The latter is \(\Omega(n^k)\) for \(k\) constant and at least
\(\Omega(n^{k-\alpha})\) for \(k = \alpha \log n\) with \(\alpha \in
\N\) constant and \(k > \alpha\). Thus the LP formulation for
\(k\)-juntas derived from the level-\(k\) Sherali-Adams hierarchy is
essentially optimal for small \(k\).
\end{example}

\section{Affine Reductions for LPs and SDPs}
\label{sec:reductions}
We will now introduce natural reductions
between problems, with control on approximation
guarantees that translate to the underlying LP and SDP level. 

\begin{definition}[Reductions between problems]
  \label{def:reduction}
  Let $\oprob_{1} = (\mathcal{S}_{1}, \mathcal{F}_{1}, \val)$ and
  $\oprob_{2} = (\mathcal{S}_{2}, \mathcal{F}_{2}, \val)$ be
  maximization problems. 
  Let \(C_{1}, S_{1}\) and \(C_{2}, S_{2}\) be guarantees
  for \(\oprob_{1}\) and \(\oprob_{2}\) respectively.
  A \emph{reduction} from $\oprob_{1}$ to $\oprob_{2}$
  respecting these guarantees
  consist of
  two maps:
\begin{enumerate}
\item $\beta \colon \mathcal{F}_{1}^{\mathcal{S}_{1}}
  \to \cone{\mathcal{F}_{2}^{\mathcal{S}_{2}}} + \R$
  rewriting instances
  as formal nonnegative combinations:
  $\beta(f_{1}) \coloneqq
  \sum_{f \in \mathcal{F}_{2}^{\mathcal{S}_{2}}} b_{f_1,f} \cdot f
  + \mu(f_1)$
  with $b_{f_1,f} \geq 0$
  for all $f \in \mathcal F_{2}^{\mathcal{S}_{2}}$;
  the term
  \(\mu(f_{1})\) is called the \emph{affine shift}
\item $\gamma \colon \mathcal{S}_1 \to \conv{\mathcal{S}_2}$
  rewriting solutions as
  formal convex combination of \(\mathcal{S}_{2}\):
  $\gamma(s_{1}) \coloneqq
  \sum_{s \in \mathcal S_2} a_{s_1,s} \cdot s$
  with $a_{s_1,s} \geq 0$ for all $s \in \mathcal S_2$ and
  $\sum_{s \in \mathcal S_2} a_{s_1,s} = 1$;
\end{enumerate}
  subject to
  \begin{equation}
    \label{eq:red-exact}
    \val_{f_{1}}(s_{1}) =
    \sum_{\substack{f \in \mathcal{F}_{2}^{\mathcal{S}_{2}} \\
        s \in \mathcal{S}_{2}}}
    b_{f_1,f} a_{s_1,s} \cdot \val_{f}(s)
    + \mu(f_1)
    ,
    \qquad s_{1} \in \mathcal{S}_{1},\
    f_{1} \in \mathcal{F}_{1}^{\mathcal{S}_{1}}
    ,
  \end{equation}
  expressing representation of the objective function of
  \(\oprob_{1}\)
  by that of \(\oprob_{2}\),
  and additionally
  \begin{equation}
    \label{eq:red-guarantee}
    C_{1}(f_{1}) \geq \sum_{f \in \mathcal{F}_{2}^{\mathcal{S}_{2}}}
    b_{f_1,f} \cdot C_{2}(f)
    + \mu(f_1)
    ,
    \qquad f_{1} \in \mathcal{F}_{1}^{\mathcal{S}_{1}},
  \end{equation}
  ensuring feasibility of the completeness guarantee.
\end{definition}

Observe that the role of soundness guarantees of
\(\oprob_1\) and \(\oprob_2\) in the definition is to
restrict the instances considered:
the map \(\beta\) involves only the instances
whose optimum value is bounded by these guarantees.
One can analogously define reductions
involving minimization problems.
E.g., for a reduction from a maximization problem \(\oprob_{1}\)
to a minimization problem \(\oprob_{2}\),
the formulas are
\begin{align*}
  \beta(f_{1}) &\coloneqq \mu(f_1) -
  \sum_{f \in \mathcal{F}_{2}^{\mathcal{S}_{2}}} b_{f_1,f} \cdot f
  \\
  \val_{f_{1}}(s_{1}) &= \mu(f_1) -
  \sum_{\substack{f \in \mathcal{F}_{2}^{\mathcal{S}_{2}} \\
      s \in \mathcal{S}_{2}}}
  b_{f_1,f} a_{s_1,s} \cdot \val_{f}(s)
  \\
  C_{1}(f_{1}) &\geq \mu(f_1) - \sum_{\val_{f} \in
  \mathcal{F}_{2}^{\mathcal{S}_{2}}} b_{f_1,f} \cdot C(f).
\end{align*}

Note that elements in $\mathcal S_2$ are obtained as convex
combinations, while elements in $\mathcal F_2$ are obtained as
nonnegative combinations and a shift.
The additional freedom for instances allows
scaling and shifting the function
values.

In a first step we will verify that a reduction between
optimization problems
$\oprob_{1}$ to $\oprob_{2}$
naturally extends to potential LP and SDP
formulations.

\begin{proposition}[Reductions of formulations]
\label{prop:probRed}
Consider a
reduction from an optimization problem $\oprob_{1}$
to another one $\oprob_{2}$
respecting completion and soundness guarantees
\(C_{1}, S_{1}\) and \(C_{2}, S_{2}\).
Then \(\fc (\oprob_{1}, C_{1}, S_{1})
\leq \fc (\oprob_{2}, C_{2}, S_{2})\)
and \(\fcp (\oprob_{1}, C_{1}, S_{1})
\leq \fcp (\oprob_{2}, C_{2}, S_{2})\).
\begin{proof}
We will use the notation from Definition~\ref{def:reduction}
for the reduction.
We only prove the claim for LP formulations
and for two maximization problems,
as the proof is analogous for SDP formulations and when
either or both problems are minimization problems.
Let us choose an LP formulation \(A x \leq b\) of \(\oprob_{2}\)
with \(x^{s}\) realizing \(s \in \mathcal{S}_{2}\)
and \(w^{f}\) realizing \(f \in \mathcal{F}_{2}^{\mathcal{S}_{2}}\).
For \(\oprob_{1}\) we shall use the same linear program
\(A x \leq b\)
with the following realizations \(y^{s_{1}}\) of
feasible solutions \(s_{1} \in \mathcal{S}_{1}\),
and \(u^{f_{1}}\) of instances
\(f_{1} \in \mathcal{F}_{1}^{\mathcal{S}_{1}}\), where
\begin{align*}
  y^{s_{1}} &\coloneqq
  \sum_{s \in \mathcal S_2} a_{s_1,s} \cdot x^{s},
  &
  u^{f_{1}}(x) &\coloneqq
  \sum_{f \in \mathcal{F}_{2}^{\mathcal{S}_{2}}} b_{f_1,f} \cdot w^{f}(x)
  + \mu(f_1)
  .
\end{align*}
As \(y^{s_{1}}\) is a convex combination of the \(x^{s}\),
obviously \(A y^{s_{1}} \leq b\).
The \(u^{f_{1}}\) are clearly affine functions with
\begin{equation*}
  u^{f_{1}}(y^{s_{1}})
  =
  \sum_{f \in \mathcal{F}_{2}^{\mathcal{S}_{2}}} b_{f_1,f} \cdot
  w^{f} \left(
    \sum_{s \in \mathcal S_2} a_{s_1,s} \cdot x^{s}
  \right)
  + \mu(f_1)
  =
  \sum_{f \in \mathcal{F}_{2}^{\mathcal{S}_{2}}} b_{f_1,f}
  \sum_{s \in \mathcal S_2} a_{s_1,s} \cdot w^{f}(x^{s})
  + \mu(f_1)
  = \val_{f_{1}}(s_{1})
\end{equation*}
by Eq.~\eqref{eq:red-exact}.
Moreover,
by Eq.~\eqref{eq:red-guarantee}.
\begin{equation*}
  \max_{A x \leq b} u^{f_{1}}(x) \leq
  \sum_{f \in \mathcal{F}_{2}^{\mathcal{S}_{2}}} b_{f_1,f} \cdot
  \max_{A x \leq b} w^{f}(x)
  + \mu(f_1)
  \leq
  \sum_{f \in \mathcal{F}_{2}^{\mathcal{S}_{2}}} b_{f_1,f} \cdot
  C_{2}(f)
  + \mu(f_1)
  \leq C_{1}(f_{1}).
  \qedhere
\end{equation*}
\end{proof}
\end{proposition}

\begin{remark}
  \label{rem:reduction-slack-matrix}
  At the level of matrices,
  Proposition~\ref{prop:probRed} can be equivalently formulated as
  follows.
  Whenever \(M_{1} = R \cdot M_{2} \cdot C + t \mathbb{1}\)
  with \(M_{1}\), \(M_{2}\), \(R\), \(C\) nonnegative matrices,
  and \(t\) a nonnegative vector,
  such that \(\mathbb{1} C = \mathbb{1}\),
  then \(\LPrk M_{1} \leq \LPrk M_{2}\) and
  \(\SDPrk M_{1} \leq \SDPrk M_{2}\).
  Note that given
  a reduction of \(\oprob_{1}\)
  to \(\oprob_{2}\)
  with the notation as in Definition~\ref{def:reduction},
  one chooses \(M_{1}\) and \(M_{2}\) to be the slack matrices
  of \(\oprob_{1}\) and \(\oprob_{2}\), respectively,
  together with matrices \(R, C\) and a vector \(t\)
  with the following entries:
  \begin{align}
    R(f_{1}, f) &= b_{f_{1}, f},
    \\
    C(s,s_{1}) &= a_{s_{1}, s},
    \\
    t(f)& = C_{2}(f_{1}) + \mu(f_1)
    - \sum_{\substack{f \in \mathcal F_2}} b_{f_1,f} \cdot C(f),
  \end{align}
  all nonnegative,
  satisfying \(M_{1} = R \cdot M_{2} \cdot C + t \mathbb{1}\).
  Now we briefly indicate a proof for this alternative formulation
  of Proposition~\ref{prop:probRed}.

  First, we prove the LP case.
  Given a size-\(r\) LP factorization
  \(M_{2} = \sum_{i \in [r]} u_{i} v_{i} + \mu \mathbb{1}\) of
  \(M_{2}\),
  one gets the size-\(r\) LP factorization
  \(M_{1} =  \sum_{i \in [r]} R u_{i} \cdot v_{i} C
  + (R \mu + t) \mathbb{1}\)
  of \(M_{1}\).

  For the SDP case,
  let \(M_{2}(i, j) = \tr[T_{i}  U_{j}] + \mu(i)\)
  be an SDP factorization of size \(r\).
  Then one can construct the following SDP factorization of \(M_{1}\)
  of size \(r\):
  \begin{equation*}
    \begin{split}
      M_{1}(f, s)
      &
      = \sum_{i,j} R(f,i) M_{2}(i,j) C(j,s) + t(f)
      = \sum_{i,j} R(f,i) (\tr[T_{i}  U_{j}] + \mu(i)) C(j,s) + t(f)
      \\
      &
      = \tr \left[
        \underbrace{\left(
            \sum_{i} R(f,i) T_{i}
          \right)}_{\widehat{T}_{f}}
        \cdot
        \underbrace{\left(
            \sum_{j} U_{j} C(j,s)
          \right)}_{\widehat{U}_{s}}
      \right]
      +
      \underbrace{\left(
          \sum_{i} R(f, i)\mu(i)  + t(f)
        \right)}_{\widehat{\mu}(f)},
   \end{split}
  \end{equation*}
  using \(\sum_{j} C(j,s) = 1\), i.e., \(\mathbb{1} C = \mathbb{1}\).
  Using the nonnegativity of \(R\), \(C\), \(\mu\), and \(t\),
  the \(\widehat{T}_{f}\) and \(\widehat{U}_{s}\) are psd,
  and the \(\widehat{\mu}(f)\) are nonnegative.
\end{remark}

For most problems, there is a natural \emph{size} \(\size{f}\)
of an instance \(f\), which is a nonnegative number,
and guarantees are often proportional to it.
In many cases there is no need to consider linear combinations
in reductions,
leading to a lightweight version of reduction of the form
\(\beta \colon \mathcal{F}_{1} \to \mathcal{F}_{2}\)
and \(\gamma \colon \mathcal{S}_{1} \to \mathcal{S}_{2}\).
For convenience, the objective function \(\val\) of
\(\oprob_{2}\) now appears on the left-hand side by rearranging.
\begin{corollary}[Inapproximability via simple reductions]
\label{cor:simpleReductions}
  Let \(\oprob_{1}\) and \(\oprob_{2}\) be two optimization problems.
  Let \(\oprob_{1}\) the completeness guarantee of \(\oprob_{1}\)
  have the form
  \begin{align*}
    C_{1}(f) &= \tau_{1} \size{f}, & f &\in \mathcal{F}_{1}
  \end{align*}
  proportional to the size \(\size{f}\) of each instance \(f\)
  with \(\size{f} \geq
  0\).  Furthermore, let
  \(\gamma \colon \mathcal{S}_{1} \rightarrow \mathcal{S}_{2}\)
  and
  \(\beta \colon \mathcal{F}_{1}^{\mathcal{S}_{1}} \rightarrow
  \mathcal{F}_{2}\)
  be maps satisfying for some constants \(\alpha, \mu\)
  \begin{align*}
    \val_{\beta(f_{1})} [\gamma(s_{1})] &=
    \alpha \val_{f_{1}}(s_{1}) + \mu \size{f_{1}}
    \\
    \size{\beta(f_{1})} &=
    (\abs{\alpha} + \mu) \cdot \size{f_{1}},
  \end{align*}
where 
  \(\alpha > 0\) if \(\oprob_{2}\) is a maximization problem,
  and
  \(\alpha < 0\) if \(\oprob_{2}\) is a minimization problem.
  Let the completeness guarantee \(C_{2}\) of \(\oprob_{2}\) be given as
  \begin{align*}
    C_{2}(f) &\coloneqq (\alpha \tau_{1} + \mu) \size{f}
    &
    f &\in \mathcal{F}_{2}
    .
  \end{align*}
  Furthermore let \(\sigma_{2}\) be a nonnegative number satisfying
  for all \(f_{1} \in \mathcal{F}_{1}^{\mathcal{S}_{1}}\)
  \begin{align*}
    \max \val_{\beta (f_{1})}
    &
    \leq
    \sigma_{2} \size{\beta(f_{1})}
    &&
    \text{if \(\oprob_{2}\) is a maximization problem}
    \\
    \min \val_{\beta (f_{1})}
    &
    \geq
    \sigma_{2} \size{\beta(f_{1})}
    &&
    \text{if \(\oprob_{2}\) is a minimization problem}
  \end{align*}
  and we set
  \begin{align*}
    S_{2}(f) &= \sigma_{2} \size{f}
    &
    f &\in \mathcal{F}_{2}
  \end{align*}
  Then \(\beta\) and \(\gamma\) form a reduction
  from \(\oprob_{1}\)
  with guarantees \(C_{1}, S_{1}\) to \(\oprob_{2}\) with
  guarantees \(C_{2}, S_{2}\).
  In particular, \(\oprob_{2}\) is inapproximable within a factor of
  \(\sigma_{2} / (\alpha \tau_{1} + \mu)\)
  by LP and SDP formulations of size less than
  \(\fc(\oprob_{1}, C_{1}, S_{1})\)
  and
  \(\fcp(\oprob_{1}, C_{1}, S_{1})\),
  respectively.
\end{corollary}

In many cases also the soundness guarantee of \(\oprob_{1}\)
is proportional to the size of \(f\), i.e., 
\begin{align*}
  S_{1}(f) &= \sigma_{1} \size{f}, & f &\in \mathcal{F}_{1}.
\end{align*}
Then a common choice is \(\sigma_{2} = \alpha \sigma_{1} + \mu\),
provided that the reduction is \emph{exact}, i.e.,
\begin{align*}
  \opt_{\oprob_{1}} \val_{\beta(f_{1})} &=
  \alpha \opt_{\oprob_{2}} \val_{f_{1}} + \mu \size{f_{1}},
\end{align*}
where the operator \(\opt_{\oprob_{1}}\) is \(\max\)
when \(\oprob_{1}\) is a maximization problem,
and the operator \(\opt_{\oprob_{1}}\) is \(\min\)
when \(\oprob_{1}\) is a minimization problem.
The operator \(\opt_{\oprob_{2}}\) is defined similarly
for \(\oprob_{2}\). We shall write \(\fc (\oprob, \tau, \sigma)\)
for \(\fc(\oprob, C, S)\)
with \(C(f) = \tau \size{f}\) and \(S(f) = \sigma \size{f}\).

The base problems \(\oprob_{1}\) from which we reduce will be
the CSPs
\(\MaxCUT\), \(\MaxCUT[\Delta]\) or
\(\MaxXOR{k}\)
in our examples.
For CSPs,
the size of an instance, i.e., weighting
\((w_{1}, \dotsc, w_{m})\)
is the total weight \(\sum_{i \in [m]} w_{i}\) of all clauses,
For 0/1 weightings representing a subset \(L\) of clauses,
the size is just the number of elements of \(L\).

The following lower bounds on formulation complexity
are implicit in \cite{CLRS13},
where similar results are written out explicitly for other constraint
satisfaction problems. The problems below constitute our base
problems and play the same role as e.g.,  $\MaxXOR{3}$ in Håstad's PCP theorem (see
\cite{haastad2001some}).

\begin{theorem}
  \label{thm:HalfXOR}
  For every \(k \geq 2\) and \(\varepsilon > 0\),
  we have $\fc(\MaxXOR{k}, 1 - \varepsilon, 1/2 + \varepsilon)$
  and  \(\fc(\MaxCUT, 1 - \varepsilon, 1/2 + \varepsilon)\)
  are both at least \(n^{\Omega(\frac{\log n}{\log \log n})}\)
  for infinitely many \(n\).
  Moreover, for the bounded degree case we have \(\fc(\MaxCUT[\Delta], 1 - \varepsilon, 1/2 + \varepsilon)
  = n^{\Omega(\frac{\log n}{\log \log n})}\)
  for infinitely many \(n\),
  where \(\Delta\) is large enough depending on \(\varepsilon\).
\begin{proof}
By \cite[Theorem~12]{Schoen-k-CSP},
for \(k \geq 3\)
there are instances \(L\) of \MaxXOR{k},
with at most \(\frac{1}{2} + \varepsilon\) of the clauses satisfiable,
but with a feasible Lasserre solution of value \(1\) in round \(\Omega(n)\).
This means in the language of \cite{CLRS13} that
\MaxXOR{k} is \((1, \frac{1}{2} + \varepsilon)\)-inapproximable
in the \(\Omega(n)\)-round Lasserre hierarchy,
and therefore also in the Sherali–Adams hierarchy.
Therefore by \cite[Theorem~3.2]{CLRS13},
\MaxXOR{k} is also
\((1 - \varepsilon, \frac{1}{2} + \varepsilon)\)-inapproximable
by an LP formulation of size \(n^{O(\log n / \log \log n)}\)
for infinitely many \(n\),
which is just a reformulation of the claim for \MaxXOR{k}.

For \(\MaxCUT\)
and \(\MaxXOR{2}\)
the argument is similar, however one uses \cite[Proof of
Theorem~5.3(I)]{CMM2009} to show
\((1 - \varepsilon, \frac{1}{2} + \varepsilon)\)-inapproximability
of \MaxCUT, and hence of \MaxXOR{2}. As the construction uses only
bounded degree graphs, where the bound \(\Delta\)
depends on \(\varepsilon\)
but not on \(n\), it follows that our argument also works for \MaxCUT[\Delta].
\end{proof}
\end{theorem}

Recall from \cite{KhotKMO07} that under the Unique Games Conjecture,
\MaxCUT cannot be approximated better than \(c_{GW}\)
by a polynomial-time algorithm.
This motivates the following conjecture, which provides our 
SDP-hard base problem. For some problems it might be
possible to also reduce from \MaxSAT{3}, which is SDP-hard
to approximate within any factor better than \(7/8\)
\cite{LRS14}.

\begin{conjecture}[SDP inapproximability of MaxCUT]
  \label{conj:SDP-base}
  For every \(\varepsilon > 0\),
  and for every constant \(\Delta\) large enough
  depending on \(\varepsilon\),
  the formulation complexity
  \(\fcp(\MaxCUT[\Delta], 1 - \varepsilon,
  c_{GW} + \varepsilon)\)
  of \MaxCUT
  is superpolynomial.
\end{conjecture}
Note however that for fixed \(\Delta\),
there are algorithms achieving an approximation factor of
\(c_{GW} + \varepsilon\) by \cite{Feige2002201},
hence in the conjecture \(\Delta\) should go to infinity
as \(\varepsilon\) tends to \(0\).

Finally, we remark that by \cite[Lemma~2.9]{GoodGW1991}
there are graphs \(G\)
where the Goemans–Williamson SDP
is off by a factor of \(c_{GW} + \varepsilon\).
For simplicity of calculations, however,
we assume for the conjecture that there are also such graphs
with SDP optimum \((1 - \varepsilon) \size{E(G)}\).

\subsection{Facial reductions and formulation complexity}
\label{sec:face-reduction}

As the notion of formulation complexity does not directly deal
with polytopes, there is no direct translation of
monotonicity of extension complexity under faces and projections (see
\cite{extform4}).
Thus many reductions that have been used
in the context of extension complexity
and polytopes do not apply, such as e.g., the one from TSP to matching
in \cite{Yannakakis88,Yannakakis91}.
Often however, the facial reduction underlies a reduction
between the problems as defined in Definition~\ref{def:reduction}.
To exemplify this we provide the underlying reduction
from TSP to matching.

\begin{example}[Maximum weight Hamiltonian cycles (uniform model)]
  \label{ex:Hamiltonian}
 We want to find a Hamiltonian cycle with maximum weight in a weighted graph.
  We consider only nonnegative weights as customary.

  Therefore for a fixed \(n\),
  we choose the \emph{feasible solutions} to be all Hamiltonian cycles
  \(C\)
  of the complete graph \(K_{n}\) on \([n]\),
  and the \emph{instances} are weighted subgraphs \(G\) of \(K_{n}\)
  with nonnegative weights.
  The objective function has the form
  \begin{equation*}
    \val_{G}(C) \coloneqq \sum_{e \in C \cap E(G)} w_{e}.
  \end{equation*}
  We shall consider the exact problem,
  i.e., with guarantees
  \(C(G) = S(G) = \max \val_{G}\).

In order to have a finite family of instances,
  one could restrict the weights to e.g., \(1\),
  essentially asking
  for the maximum number of edges
  a Hamiltonian cycle can have in common with a given subgraph.
  For the following reduction, we will use weights
  \(\face{1,2}\) and we adapt Yannakakis's construction
  to reduce the
  maximum matching problem on \(K_{2n}\)
  to the maximum weight Hamiltonian cycle problem on \(K_{4n}\).
  To simplify notation, we identify \([4n]\) with
  \(\{0,1\} \times [2n]\),
  i.e., the vertices are labelled by pairs \((i, j)\) with \(i\)
  being \(0\) or \(1\) and \(j \in [2n]\).
  Given a graph \(G\) on \([2n]\),  we think of it as being
  supported on \(\{0\} \times [2n]\).

  We consider the weighted graph \(\widetilde{G}\)
  with edges and weights:
  \begin{center}
    \begin{tabular}{ccc}
      Edge & Weight \\
      \hline
      \(\{(0,j), (1, j)\}\) & \(2\) & \(j \in [2n]\) \\
      \(\{(1, j), (1, k)\}\) & \(1\) & \(j, k \in [2n]\) \\
      \(\{(0, j), (0, k)\}\) & \(1\) & \(\{j,k\} \in E(G)\)
    \end{tabular}
  \end{center}
  For every perfect matching \(M\) on \([2n]\), choose a Hamiltonian
  cycle \(C_{M}\) containing the edges
  \begin{enumerate}
  \item \(\{(0,j), (1, j)\}\) for \(j \in [2n]\),
  \item \(\{(0, j), (0, k)\}\) for \(\{j,k\} \in E(G)\),
  \item \(n\) additional edges of the form \(\{(1, j), (1, k)\}\)
    to obtain a Hamiltonian cycle.
  \end{enumerate}
  Note that
  \begin{equation}
    \label{eq:TSP-matching-weight}
    \val_{\widetilde{G}}(C_{M}) = 5n + \size{M \cap E(G)}.
  \end{equation}

  We now determine the maximum of \(\val_{\widetilde{G}}\)
  on all Hamiltonian cycles.
  Therefore let \(C\) be an arbitrary Hamiltonian cycle.
  Let us consider \(C\) restricted to \(\{0\} \times [2n]\);
  its components are (possible empty) paths.
  Let \(k\) be the number of components,
  which are non-empty paths, and contained in \(G\).
  Obviously, \(k \leq \nu(G)\), where \(\nu(G)\) is the matching
  number,
  as selecting one edge from every such component
  provides a \(k\)-matching of \(G\).

  Let \(l\) be the number of components containing
  at least one edge not in \(G\).
  Note that \(k + l \leq n\),
  because choosing one edge of all these \(k+l\)
  components, we obtain a \((k+l)\)-matching on \([2n]\),
  similarly as in the previous paragraph.

  Finally, let \(m\) be the number of single vertex components.
  Therefore \(C\) contains exactly \(2n - (k+l+m)\)
  edges on \(\{0\} \times [2n]\), of which at least \(l\)
  are not contained in \(G\).
  Hence the contribution of these edges
  to the weight \(\val_{\widetilde{G}}(C)\) is at most
  \begin{equation}
    \label{eq:wG-C-0}
    \val_{\widetilde{G}}(C \cap E(\{0\} \times [2n]))
    \leq 2n - (k+l+m) - l = 2n - k - 2l - m.
  \end{equation}
  Moreover, the cycle \(C\) contains exactly \(2n - (k+l+m)\) edges
  on \(\{1\} \times [2n]\) whose contribution to the weight is
  \begin{equation}
    \label{eq:wG-C-1}
    \val_{\widetilde{G}}(C \cap E(\{1\} \times [2n])) = 2n - (k+l+m).
  \end{equation}
  Finally, \(C\) contains \(2 (k+l+m)\) edges between the partitions
  \(\{0\} \times [2n]\) and \(\{1\} \times [2n]\),
  all of which have weight at most \(2\).
In fact, at each of the \(m\) single vertex components
  in \(\{0\} \times [2n]\), only one of the edges can be of the form
  \(\{(0, j), (1, j)\}\),
  the other edge must have weight \(0\).
  Therefore the contribution of the edges between the partitions
  is at most
  \begin{equation}
    \label{eq:wG-C-0-1}
    \val_{\widetilde{G}}(C \cap E(\{0\} \times [2n],
    \{1\} \times [2n])
    \leq
    2 [2(k+l+m) - m] = 4k + 4l + 2m.
  \end{equation}
  Summing up Eqs.~\eqref{eq:wG-C-0}, \eqref{eq:wG-C-1}
  and \eqref{eq:wG-C-0-1},
  we obtain the following upper bound on the weight of \(C\):
  \begin{equation*}
    \val_{\widetilde{G}}(C) \leq 4n + 2k + l \leq 5n + k
    \leq 5n + \nu(G).
  \end{equation*}
  Together with Eq.~\eqref{eq:TSP-matching-weight},
  this proves \(\max_{C} \val_{\widetilde{G}}(C) = 5n + \nu(G)\).
  Thus the \(\val_{\widetilde{G}}\) and \(C_{M}\)
  reduce the maximum matching problem
  to the maximum Hamiltonian cycle problem with
  \(\beta(G) = \widetilde{G}\), \(\gamma(M) = C_{M}\)
  and \(\mu(G) = 5n\).
  Hence the LP formulation complexity of
  the maximum Hamiltonian cycle problem is \(2^{\Omega(n)}\) by
  Proposition~\ref{prop:probRed}.
\end{example}

\section{Inapproximability of \VertexCover and \MaxIndep}
\label{sec:vertex-cover}

We will now establish inapproximability results for \VertexCover and
\MaxIndep via reduction from \MaxCUT, even for bounded degree
subgraphs.  These two problems are of particular interest, answering a question of
\cite{Singh10} and \cite{CLRS13} as well as a weak version of sparse
graph conjecture from \cite{BFP2013}. Moreover, \VertexCover is not of
the CSP
type, therefore the framework in \cite{CLRS13} does not apply.
Using our reduction framework,
recently these results have been further improved in
\cite{VertexCover2015} to obtain
\((2 - \varepsilon)\)-inapproximability for \VertexCover (which is optimal)
and inapproximability of \MaxIndep
within any constant factor. 

Formulation complexity depend heavily on
how a problem is formulated.
For example, the model of \MaxIndep used here is
motivated by its combinatorial counterpart,
and captures standard LPs,
like the ones coming from Sherali–Adams hierarchies.
In this model,
\MaxIndep for a given graph \(G\) is approximable within a factor of
\(2 \sqrt{n}\) with a polynomial sized LP, see \cite{VertexCover2015}.
However, the formulation complexity of \emph{another model of} the
maximum independent set
problem with an approximation factor \(n^{1-\varepsilon}\)
is subexponential, rephrasing \cite{extform4} (see also
\cite{bfps2012,braverman2012information,bfps2012jour}), see
Section~\ref{sec:stable-set-problem}.
In this model the instances come from the polytope world,
and are actually formal linear combinations of several graphs,
and this makes the difference.

The current best PCP bound for bounded degree \MaxIndep can be found
in \cite{chan13._resist}. See also
\cite{austrin2009inapproximability} for inapproximability results
assuming the Unique Games Conjecture.

The minimization problem \(\VertexCover{G}\)
of a graph \(G\)
asks for a minimum weighted vertex cover of $G$.
We consider the non-uniform model with
instances being the induced subgraphs of $G$.

\begin{definition}[\VertexCover]
  \label{def:vertex-cover}
  Given a graph \(G\), the problem \(\VertexCover{G}\)
  has all vertex covers \(S\) of \(G\) as feasible solutions,
  and instances all induced subgraphs \(H\) of \(G\).
  The problem \(\VertexCover{G}\) is the minimization problem
  with its objective function having values
  \(\val_{H}(S) \coloneqq \size{S \cap V(H)}\).
  The problem \(\VertexCover[\Delta]{G}\)
  is the restriction of instances to induced subgraphs \(H\),
  with maximum degree at most \(\Delta\).
\end{definition}
Note that for every vertex cover \(S\) of \(G\),
any induced subgraph \(H\) has \(S \cap V(H)\) as a vertex cover,
and all vertex covers of \(H\) are of this form.
In particular, \(\min \val_{H}\) is the minimum size of
a vertex cover of \(H\).

The problem \MaxIndep asks for maximum sized independent sets
in graphs.
As independent sets are exactly the complements of vertex covers,
it is natural to use a formulation similar to \VertexCover.
\begin{definition}[\MaxIndep]
  \label{def:max-indep}
  Given a graph \(G\), the maximization problem \(\MaxIndep{G}\)
  has all independent sets \(S\) of \(G\) as feasible solutions,
  and instances are all induced subgraphs \(H\) of \(G\).
  The objective function is
  \(\val_{H}(S) \coloneqq \size{S \cap V(H)}\).
  The subproblem \(\MaxIndep[\Delta]{G}\)
  is the restriction to all induced subgraphs \(H\)
  with maximum degree at most \(\Delta\).
\end{definition}

For both \VertexCover and \MaxIndep,
we shall use the following conflict graph \(G\)
for a fixed \(n\),
similar to \cite{feige91._approx_np}; we might think of \(G\)
as a \emph{universal} graph encoding all possible
instances. 
Let the vertices of \(G\) be
all partial assignments \(\sigma\) of
two variables \(x_{i}\) and \(x_{j}\)
satisfying the 2-XOR clause \(x_{i} \oplus x_{j} = 1\).
Two vertices \(\sigma_{1}\) and \(\sigma_{2}\)
are connected if and only if
the assignments \(\sigma_{1}\) and \(\sigma_{2}\)
are incompatible
(i.e., assign different truth values to some variable),
see Figure~\ref{fig:VC-graph} for an illustration.
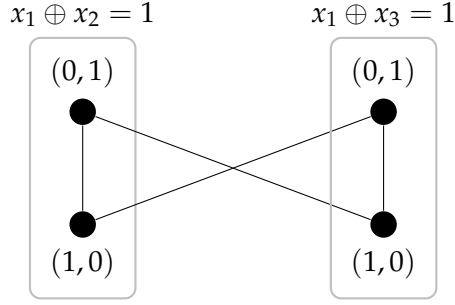
\begin{figure}
  \centering
  \begin{tikzpicture}
    \node (1=1-2=0) at (0,0)
    [vertex, label={[name={label-1=1-2=0}]below:{\((1,0)\)}}]{};
    \node (1=0-2=1) at (0,1.5)
    [vertex, label={[name={label-1=0-2=1}]above:{\((0,1)\)}}]{}
    edge (1=1-2=0);
    \node[container,
    fit = {(1=1-2=0) (label-1=1-2=0) (1=0-2=1) (label-1=0-2=1)},
    label=above:{\(x_{1} \oplus x_{2} = 1\)}]{};
    \node (1=1-3=0) at (4,0)
    [vertex, label={[name={label-1=1-3=0}]below:{\((1,0)\)}}]{}
    edge (1=0-2=1);
    \node (1=0-3=1) at (4,1.5)
    [vertex, label={[name={label-1=0-3=1}]above:{\((0,1)\)}}]{}
    edge (1=1-3=0)
    edge (1=1-2=0);
    \node[container,
    fit = {(1=1-3=0) (label-1=1-3=0) (1=0-3=1) (label-1=0-3=1)},
    label=above:{\(x_{1} \oplus x_{3} = 1\)}]{};
  \end{tikzpicture}
  \caption{Conflict graph of 2-XOR clauses. We include edges between
    all conflicted partial assignment to variables.}
  \label{fig:VC-graph}
\end{figure}
As we are considering problems for optimizing size of vertex sets,
it is natural to define the size of an instance,
i.e., a subgraph \(K\),
as the size of its vertex set \(\size{V(K)}\).

\begin{theorem}
  \label{thm:VertexCover}
  For every \(\varepsilon > 0\) there is a \(\Delta\)
  such that
  for infinitely many \(m\),
  there is a graph \(G\) with \(\size{V(G)} = m\)
  such that
  \(\fc(\VertexCover[\Delta]{G}, 1/2 + \Theta(\varepsilon),
  3/4 - \Theta(\varepsilon))
  \geq m^{\Omega(\frac{\log m}{\log \log m})}\)
  showing inapproximability within a factor of
  \(\frac{3}{2} - \Theta(\varepsilon)\),
  and also
  \(\fc(\MaxIndep[\Delta]{G}, 1/2 - \Theta(\varepsilon),
  1/4 + \Theta(\varepsilon))
  \geq m^{\Omega(\frac{\log m}{\log \log m})}\)
  establishing inapproximability factor
  \(\frac{1}{2} + \Theta(\varepsilon)\).
  Assuming Conjecture~\ref{conj:SDP-base},
  we also have
  \(\fcp(\VertexCover[\Delta]{G}, 1/2 + \Theta(\varepsilon),
  1 - c_{GW} / 2 - \Theta(\varepsilon))\)
  and
  \(\fcp(\MaxIndep[\Delta]{G},  1/2 - \Theta(\varepsilon),
  c_{GW} / 2 + \Theta(\varepsilon))\)
  are superpolynomial,
  achieving inapproximability factors
  \(2 - c_{GW} - \Theta(\varepsilon)\)
  and \(c_{GW} + \Theta(\varepsilon)\), respectively.
\begin{proof}
We shall use the graph \(G\) constructed above,
which has \(m = 2 \binom{n}{2}\) vertices.
We reduce \(\MaxCUT[\Delta]\) to \(\VertexCover[2 \Delta - 1]{G}\)
using Corollary~\ref{cor:simpleReductions} with
\(\alpha = -1\), \(\mu = 2\), \(\tau_{1} = 1 - \varepsilon\),
\(\sigma_{1} = 1/2 + \varepsilon\).
For demonstration purposes, we shall write out the explicit
guarantees below.
Recall that for a graph \(K\) on \([n]\)
with maximum degree at most \(\Delta\),
the guarantees for \MaxCUT are
\(C_{\MaxCUT}(K) = (1 - \varepsilon) \size{E(K)}\)
and \(S_{\MaxCUT}(K) = (1/2 + \varepsilon) \size{E(K)}\).
For \VertexCover[2 \Delta - 1]{G}, we have the following explicit guarantees:
\begin{align*}
  C_{\VertexCover{G}} (H) &= (1/2 + \varepsilon/2) \size{V(H)}
  ,
  \\
  S_{\VertexCover{G}} (H) &= (3/4 - \varepsilon/2) \size{V(H)}
  .
\end{align*}

Let \(H(K)\) be the induced subgraph of \(G\) on the set
\(V(H(K)) \coloneqq \set{\sigma}{\{i,j\} \in E(K),\ \dom \sigma
  =\{x_{i}, x_{j}\} }\)
of all partial assignments \(\sigma\)
which assign values to variables \(x_{i}\), \(x_{j}\)
corresponding to an edge \(\{i, j\}\) of \(K\).
In particular, \(\size{V(H(K))} = 2 \size{E(K)}\),
as there are two partial assignments per each edge \(\{i, j\}\).

Note that for every partial assignment \(\sigma\)
to \(x_{i}\) and \(x_{j}\),
there are \(2 \Delta - 1\) partial assignments incompatible with it
in \(V(H(K))\):
exactly one assignment for every edge of \(K\)
incident to \(i\) or \(j\).
Thus the maximum degree of \(H(K)\) is at most \(2 \Delta - 1\).

We now define the two maps providing the reduction. Let \(\beta(K) \coloneqq H(K)\).
For a total assignment \(s\),
let \(\gamma(s) \coloneqq \set{\sigma}{\sigma \nsubseteq s}\)
be the set of partial assignments incompatible with \(s\);
this is clearly a vertex cover. 

It remains to show that this is a reduction.
For every edge \(\{i, j\} \in K\),
there are two partial assignments \(\sigma\)
to \(x_{i}\) and \(x_{j}\)
satisfying \(x_{i} \oplus x_{j} = 1\).
If \(s\) satisfies \(x_{i} \oplus x_{j} = 1\),
i.e., \(\{i, j\}\) is in the cut induced by \(s\),
then exactly one of the \(\sigma\)
is compatible with \(s\),
otherwise both of the assignments are incompatible.
This provides
\[
\val^{\VertexCover}_{H(K)} [\gamma(s)] =
\size{\set{\sigma}{\sigma \nsubseteq s}}
= 2 \size{E(K)} - \val^{\MaxCUT}_{K}(s)
.\]
To compare optimum values,
note that for any vertex cover \(S\) of \(G\),
the partial assignments
\(\set{\sigma}{\sigma \notin S}\)
occurring in the complement of \(S\) are compatible (as the complement
forms a stable set),
hence there is a global assignment \(s\) of \(x_{1}, \dotsc, x_{n}\)
compatible with all of them.
In particular, \(\gamma(s) \subseteq S\),
hence \(\val^{\VertexCover}_{H(K)} (S)
\geq \val^{\MaxCUT}_{K} [\gamma(s)]\),
so that we obtain
\begin{equation*}
 \begin{split}
  \min \val^{\VertexCover}_{H(K)} &=
  \min_{s} \val^{\VertexCover}_{H(K)} [\gamma(s)]
  = 2 \size{E(K)} - \max \val^{\MaxCUT}_{K}
  \\
  &
  \geq 2 \size{E(K)} - (1/2 + \varepsilon) \size{E(K)}
  = (3/4 - \varepsilon/2) \size{V(H(K))}
  = S_{\VertexCover}(H(K))
  .
 \end{split}
\end{equation*}
Finally, it is easy to verify that
\(C_{\VertexCover}(H(K)) = 2 \card{E(K)} - C_{\MaxCUT}(K)\).
This finishes the proof that \(\beta\) and \(\gamma\)
define a reduction to \VertexCover[\Delta]{G}.
Hence by Corollary~\ref{cor:simpleReductions}  (or
Proposition~\ref{prop:probRed}), using \(m = 2 \binom{n}{2}\),
we have 
\(\fc(\VertexCover[2 \Delta - 1]{G},3/2 - 2
\varepsilon,1 + 2\varepsilon)
= n^{\Omega(\frac{\log n}{\log \log n})}
= m^{\Omega(\frac{\log m}{\log \log m})}\)
for infinitely many \(n\).

For \MaxIndep,
we apply
a similar reduction from \MaxCUT[\Delta]
to \MaxIndep[2 \Delta - 1]{G}. 
We define \(\beta(K) \coloneqq H(K)\) as above and we
set 
\(\gamma(s) \coloneqq \set{\sigma}{\sigma \subseteq s}\)
to be the set of partial assignments compatible with
the total assignment \(s\), this is clearly an independent set,
containing exactly one vertex per satisfied clause.
In particular, \(\val_{H(K)} [\gamma(s)] = \val_{K}(s)\).
The rest of the argument is analogous to the case of
\(\VertexCover[\Delta]{G}\), and hence omitted. 
Now the parameters for Corollary~\ref{cor:simpleReductions}
are
\(\alpha = 1\), \(\mu = 0\), \(\tau_{1} = 1 - \varepsilon\),
and
\(\sigma_{1} = 1/2 + \varepsilon\).

The SDP
inapproximability factors follow similarly, by replacing
\(\frac{1}{2}\) with \(c_{GW}\).
\end{proof}
\end{theorem}

\section{Inapproximability of CSPs}
\label{sec:futh-hardn-appr}

In this section we present example reductions for minimum and maximum
constraint satisfaction problems.  The results for \emph{binary
  \problem{Max-CSP}s},
(for CSPs as defined in Definition~\ref{def:maxCSP})
 could also be obtained in the LP case from
\cite{CLRS13} by combination with the respective
Sherali–Adams/Lasserre gap instances \cite{James2014}.  For simplicity
of exposition, we reduce from \(\MaxXOR{2}\),
or sometimes \MaxCUT, however by reducing from the subproblem
\MaxCUT[\Delta], we immediately obtain the results for bounded
occurrence of literals, with \(\Delta\)
depending on the approximation factor.

\subsection{\MaxMULTICUT{k}: a non-binary CSP}
\label{sec:problemmax-multi-k}

The \MaxMULTICUT{k} problem is interesting on its own
being a CSP over a non-binary alphabet, thus the framework in
\cite{CLRS13} does not readily apply.
Note that \MaxMULTICUT{k} is APX-hard, as it contains \MaxCUT.
The current best PCP inapproximability bound
\(1 - 1/(34k) + \varepsilon\)
is given by \cite{MaxkCUT1997}.

Here we omit the definition of non-binary CSPs,
where the feasible solutions are no longer two-valued
assignments, and restrict to \MaxMULTICUT{k}.

\begin{definition}[\MaxMULTICUT{k}]
	For fixed positive integers \(n\) and \(k\), the problem 
	\MaxMULTICUT{k} has
   \begin{enumerate}
   \item \textbf{feasible solutions:}
      all partitions of \([n]\) into \(k\) sets;
   \item \textbf{instances:}
		all graphs \(G\) with \(V(G) \subseteq [n]\).
  \item \textbf{objective function:}
    for a graph \(G\) and a partition \(p\) of \([n]\),
    let \(\val_{G}(p)\) be the number of edges of \(G\)
    whose end points lie in different cells of \(p\).
   \end{enumerate}
\end{definition}
This differs from a binary CSP only by having a different kind of
feasible solutions.
Hence it is still natural to define the size of an instance,
i.e., graph \(G\), as the number of clauses, i.e,
number of edges \(\size{E(G)}\).

\begin{corollary}
	Let \(k\ge 3\) be a fixed integer. Then for infinitely 
	many \(n\),
  \[\fc(\MaxMULTICUT{k}, c(k) + 1 - \varepsilon,
  c(k) + 1/2 + \varepsilon) \ge
	n^{\Omega \left( \frac{\log n}{\log\log n} \right)}\]
  achieving inapproximability factor
  \(\frac{2 c(k) + 1}{2 c(k) + 2} + \Theta(\varepsilon)\),
  where
  \[c(k) \coloneqq
	\binom{k-2}{2}\left(\binom{k+2}{2}-3\right) + 
	2(k-2)\left(\binom{k+2}{2}-3\right).\]
      Assuming Conjecture~\ref{conj:SDP-base},
      we also have that
      \(\fcp(\MaxMULTICUT{k}, c(k) + 1 - \varepsilon,
      c(k) + c_{GW} + \varepsilon)\)
      is superpolynomial,
      showing inapproximability factor
      \[\frac{c(k) + c_{GW}}{c(k) + 1} + \Theta(\varepsilon).\]
\begin{proof}
We reduce \(\MaxCUT\) to \MaxMULTICUT{k}.
The reduction is essentially
        identical to 
	\cite{papadimitriou1991optimization}, however we have to verify its
        compatibility with our reduction mechanism.  To this end it
        will suffice to define the reduction maps \(\beta\) and
        \(\gamma\).

  Given a graph \(G\),
  we construct a new graph \(\beta(G)\)
  as illustrated in Figure~\ref{fig:MultiCUT}.
  \begin{figure}
    \centering
    \begin{tikzpicture}[vertex1/.style={vertex, fill=MidnightBlue},
      vertex2/.style={vertex, fill=ForestGreen},
      vertex3/.style={vertex, fill=Maroon},
      x=5em, y=5em, cell/.style={fill=lightgray}, very thick]
      \node(1)[vertex1, big vertex, label={[name=label-1]above:{1}},
        at={(0,2)}]{};
      \node(2)[vertex2, big vertex, label={[name=label-2]below:{2}},
        at={(0,0)}]{}
	edge[dashed] (1);
      \node[container, fit=(1) (2) (label-1) (label-2),
        label=below:{\(G\)}] {};
      \node(3)[vertex3, label=left:{-3}, at={(1,1)}]{};
      \node(1-2-1)[vertex1, at={(-2,0)},
      label=below:{\(( 1, 2; 1; (1))\)}]{};
      \node(1-2-2)[vertex2, at={(1-2-1 |- 1)},
      label=above:{\(( 1, 2; 1; (2))\)}]{};
      \node(1-2-3)[vertex3, at={(-1,1)},
      label=left:{\(( 1, 2; 1; 1)\)}]{};
      \foreach \x/\y in {1/2,1/3, 2/3}
      {\draw(1-2-\x) -- (1-2-\y);
        \draw (\x) -- (1-2-\y);}
      \draw (2) -- (1-2-1);
      \begin{scope}[on background layer]
      \fill[cell] (1.center) -- (2.center) -- (1-2-1.center)
      -- (1-2-2.center) --cycle;
      \end{scope}
      \node(1-3-1)[vertex1, at={(2,2)},
      label=right:{\(( 1, 3; 1; (1))\)}]{};
      \node(1-3-2)[vertex2, at={(1 -| 3)},
      label=below:{\(( 1, 3; 1; 1)\)}]{};
      \node(1-3-3)[vertex3, at={(1,3)},
      label={\(( 1, 3; 1; (-3))\)}]{};
      \foreach \x/\y in {1/2,1/3, 3/2}
      {\draw(1-3-\x) -- (1-3-\y);
        \draw (\x) -- (1-3-\y);}
      \draw (3) -- (1-3-1);
      \begin{scope}[on background layer]
      \fill[cell] (1.center) -- (3.center) -- (1-3-1.center)
      -- (1-3-3.center) --cycle;
      \end{scope}
      \node(2-3-1)[vertex1, at={(1,0)},
      label=above:{\((2, 3; 1; 1)\)}]{};
      \node(2-3-2)[vertex2, at={(2,0)},
      label=right:{\((2, 3; 1; (2))\)}]{};
      \node(2-3-3)[vertex3, at={(1,-1)},
      label=below:{\((2, 3; 1; (-3))\)}]{};
      \foreach \x/\y in {2/1,3/1, 3/2}
      {\draw(2-3-\x) -- (2-3-\y);
        \draw (\x) -- (2-3-\y);}
      \draw (2) -- (2-3-3);
      \begin{scope}[on background layer]
      \fill[cell] (2.center) -- (3.center) -- (2-3-2.center)
      -- (2-3-3.center) --cycle;
      \end{scope}
    \end{tikzpicture}
    \caption{\label{fig:MultiCUT}
      Reduction between \MaxCUT and \MaxMULTICUT{k}.
      Here \(k = 3\) and \(G = K_{2}\).
      The dashed edge denotes the edge of \(G\),
      which is \emph{not} contained in the reduction.
      The squares are the copies of almost complete graphs added.
      The partition is represented by coloring the vertices:
      blue and green are the original cells on \([2]\),
      and red is an additional color.}
  \end{figure}
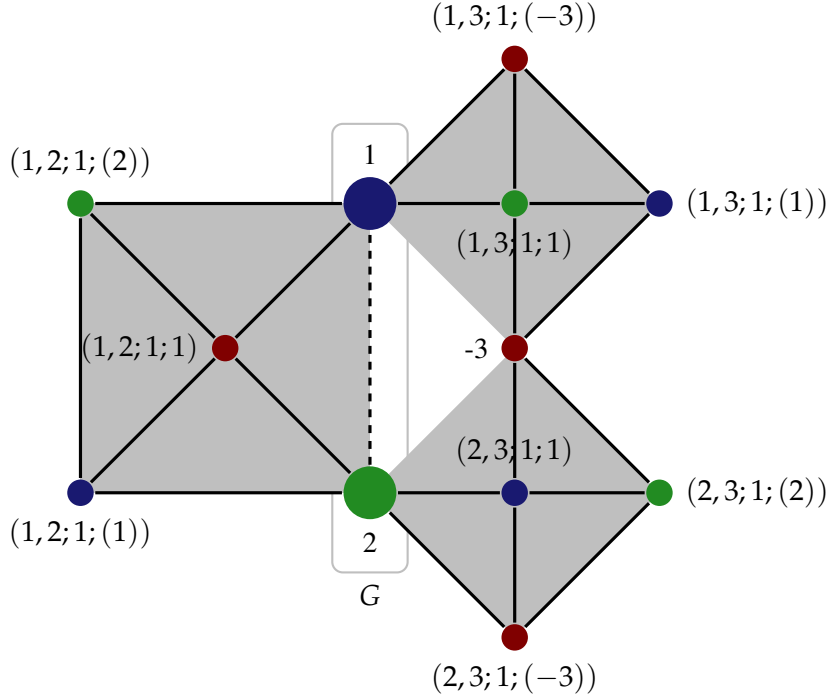
  Consider the vertices of \(G\)
  together with \(k - 2\) new vertices \(-k, \dotsc, -3\).
  For every pair of vertices \(i, j\)
  we add \(n_{ij}\) copies of an almost complete graph
  on \(k + 2\) vertices, two of which are \((i), (j)\), as follows.

  First, let us determine the number \(n_{ij}\) of copies added.
  For \(i, j \in [n]\), we add one copy if \(i\) and \(j\)
  are connected in \(G\), and no copies otherwise.
  For \(i, j \in \{-k, \dotsc, -3\}\)
  we add \(\size{E(G)}\) many copies.
  Finally, for \(i \in [n]\) and \(j \in \{-k, \dotsc, -3\}\)
  we add \(\deg_{G}(i)\) many copies.

  Now let us describe the copies themselves.
  Let us fix \(i, j\) and let \(x \in [n_{ij}]\)
  be an index of the copy.
  We add \(k\) new vertices \((i,j; x; (i))\), \((i,j; x; (j))\),
  and \((i,j; x; t)\) for \(t = 1, \dotsc, k-2\).
  We connect every pair of the \(k +2 \) vertices
  \(i, j, (i,j; x; (i)), (i,j; x; (j)), (i, j; x; t)\)
  with an edge except the pairs
  \(\{i, (i,j; x, (i))\}\),
  \(\{i, j\}\),
  and
  \(\{(i,j; x, j), (j)\}\).
  This is done for all \(i, j, x\), and
  we let \(\beta(G)\) be the graph so obtained.

  By construction, \(\beta(G)\) has
  \begin{equation*}
    \size{E(\beta(G))} =
    \left[
      \binom{k + 2}{2} - 3
    \right]
    \cdot
    \sum_{ij} n_{ij}
    =
    \left[
      \binom{k + 2}{2} - 3
    \right]
    \left(
      1 + 2 (k - 2) + \binom{k - 2}{2}
    \right)
    \size{E(G)}
  \end{equation*}
  many edges, and its vertex set
  \(V(\beta(G))\) is contained in the set
  \begin{equation*}
    [n] \cupdot \{-k, \dotsc, -3\} \cupdot
    \set{(i, j; x; t)}{-k \leq i < j \leq 3, x \in [n_{ij}],
      t \in [k]}
  \end{equation*}
  having size polynomial in \(n\):
  \begin{equation*}
    m \coloneqq
    n + (k - 2)
    + k \left[\binom{k-2}{2} + 2 (k-2) + 1\right] \binom{n}{2}.
  \end{equation*}
  vertices.

  We define \(\gamma\) to map every \(2\)-partition \(p\) of \([n]\)
  into a \(k\)-partition of \(\beta(G)\) by extending it as follows.
  Let \(p_{1}\) and \(p_{2}\) denote the cells of \(p\).
  Elements of \(p_{1}\) and \(p_{2}\) go to the first and second
  cell of the \(\gamma(p)\), respectively.
  The new vertices \(-i\) added for \(i = -k, \dotsc, -3\)
  go to the \(i\)-th cell.
  Vertices \((i,j; x; (i))\) and \((i,j; x; (j))\)
  go to the cell of \(i\) and \(j\), respectively.
  For fixed \(i, j, x\)
  the vertices \((i,j; x; t)\) for \(t = 1, \dotsc, k - 2\)
  are put into \(k - 2\) different cells,
  which do not contain \((i,j; x; (i))\) or \((i,j; x; (j))\).
  This is possible as there are \(k\) different cells.

        By \cite{papadimitriou1991optimization},
	\begin{align*}
    \val^{\MaxMULTICUT{k}}_{\beta(G)} [\gamma(p)]
    &= \val^{\MaxCUT}_{G}(p) + \mu(G), \\
    \max \val^{\MaxMULTICUT{k}}_{\beta(G)}
    &= \max \val^{\MaxCUT}_{G} + \mu(G),
    \intertext{where}
    \mu(G) &=
    \size{E(\beta(G))} - \size{E(G)} =
    \underbrace{\left[
        \binom{k-2}{2} + 2(k-2) \right]
      \left[ \binom{k+2}{2}-3 \right]}_{c(k)}
              \size{E(G)}
    .
	\end{align*}
  Therefore we obtain a reduction from \MaxCUT on \(n\) vertices
  to \MaxMULTICUT{k} on \(m\) vertices,
  with \(m\) polynomially bounded in \(n\). Combining
  Corollary~\ref{cor:simpleReductions} with Theorem~\ref{thm:HalfXOR}
  in the LP case
  with parameters 
  \(\alpha = 1\), \(\mu = c(k)\), \(\tau_{1} = 1 - \varepsilon\),
  \(\sigma_{1} = 1/2 + \varepsilon\),
  we obtain that
	\(\fc(\MaxMULTICUT{k}, c(k) + 1 - \varepsilon,
  c(k) + 1/2 + \varepsilon)\)
  and
  is \(m^{\Omega(\frac{\log m}{\log\log m})}
  = n^{\Omega(\frac{\log n}{\log\log n})}\).
The SDP case follows analogously from Conjecture~\ref{conj:SDP-base}
using \(\sigma_{1} = c_{GW}/2 + \varepsilon\).
\end{proof}
\end{corollary}

\subsection{Inapproximability of general 2-CSPs}
\label{sec:other-2-csps}

First we consider general CSPs with no restrictions on constraints,
for
which the exact approximation factor can be easily established. We
present the hardness of LP approximation here. The LP with matching
factor can be found in \cite{trevisan98._parallelLP}.

\begin{definition}[\MaxCSP{2} and \MaxCONJSAT{2}]
  The problem \(\MaxCSP{2}\) is the CSP on variables
  \(x_1,\dots,x_n\) with constraint family
  \(\mathcal{C}_{2CSP}\)
  consisting of all possible constraints
  depending on at most two variables.
  The problem \(\MaxCONJSAT{2}\) is the CSP
  with constraint family consisting of all
  possible conjunctions of two literals.
\end{definition}

\begin{corollary}
  \label{cor:MaxCSP}
  For every \(\varepsilon > 0\) and infinitely many \(n\),
  we have
  \(\fc(\MaxCSP{2}, 1 - \varepsilon, 1/2 + \varepsilon) \ge
  n^{\Omega(\frac{\log n}{\log\log n})}\) achieving inapproximability
  factor \(\frac{1}{2} + \Theta(\varepsilon)\),
  where \(n\) is the number of variables
  of \(\MaxCSP{2}\).
  Similarly,
  \(\fc(\MaxCONJSAT{2}, 1/2 - \varepsilon, 1/4 + \varepsilon)
  \geq n^{\Omega(\frac{\log n}{\log\log n})}\)
  establishing inapproximability
  factor \(\frac{1}{2} + \Theta(\varepsilon\))
  for infinitely many \(n\). Moreover, in the SDP case, assuming Conjecture~\ref{conj:SDP-base}, we have 
  \(\fcp(\MaxCSP{2}, 1 - \varepsilon, c_{GW} + \varepsilon)\)
  and
  \(\fcp(\MaxCONJSAT{2}, 1/2 - \varepsilon, c_{GW}/2 + \varepsilon)\)
  are superpolynomial
  showing inapproximability factor \(c_{GW} + \Theta(\varepsilon)\).
\begin{proof}
We identify \(\MaxXOR{2}\) as a subproblem of \(\MaxCSP{2}\):
Every \(2\)-XOR clause is evidently a boolean function of \(2\)
variables.
So restricting the instances
of \(\MaxCSP{2}\) to \(2\)-XOR clauses
with 0/1 weights gives \(\MaxXOR{2}\).
Now with Theorem~\ref{thm:HalfXOR},
the result follows.

The claim about \MaxCONJSAT{2} follows via the reduction
from \MaxCSP{2} to \MaxCONJSAT{2}
in \cite{trevisan98._parallelLP}.
We prefer to reduce from \MaxXOR{2} instead
for easier control over the approximation guarantees.
The idea is to write each clause \(C\)
in disjunctive normal form,
and replace \(C\) with the set \(S(C)\) of conjunctions
in its normal form,
one conjunction for every assignment satisfying \(C\).
In particular, for 2-XOR clauses
\(S(x_{i} \oplus x_{j} = 1) =
\{x_{i} \land \lnot x_{j}, \lnot x_{i} \land x_{j}\}\)
and \(S(x_{i} \oplus x_{j} = 0) =
\{x_{i} \land x_{j}, \lnot x_{i} \land \lnot x_{j}\}\).
Therefore formally,
a set of clauses \(L\)
is mapped to
\(\beta(L) = \bigcup_{C \in L} S(C)\).
Every assignment of variables is mapped to themselves,
i.e., \(\gamma\) is the identity.
We have \(\val_{\beta(L)} (s) = \val_{L}(s)\)
and \(\size{\beta(L)} = 2 \size{L}\).
Now the claim follows.
\end{proof}
\end{corollary}

\subsection{\MaxSAT{2} and \MaxSAT{3} inapproximability}

We now establish an LP-inapproximability factor of
\(\frac{3}{4}+ \varepsilon\) for \MaxSAT{2} via a direct reduction
from \MaxCUT and an SDP-inapproximability factor of about
\(0.93928+\varepsilon\) assuming Conjecture~\ref{conj:SDP-base}.
Note that
\cite{goemans94._maxSAT} show the existence of an LP that achieves a factor of
\(\frac{3}{4}\), so that our estimation is tight in the LP
case. Moreover, in \cite{feige1995approximating} it is shown that
\MaxSAT{2} can be approximated with a small SDP within a factor of
0.931 leaving a (conditional) gap of about 0.08.

Obviously, the same factor applies for \(\MaxSAT{k}\) with \(k
\geq 2\), too.
Note that we allow clauses with less than \(k\) literals in
\(\MaxSAT{k}\), which is
in line with the definition in \cite{Schoen-k-CSP} to maintain
compatibility.
Note that \cite[Theorem~1.5]{LRS14} establishes
\(7/8 + \varepsilon\) inapproximability for \MaxSAT{3}
even in the SDP case.

\begin{definition}[\MaxSAT{k}]
	For fixed \(n,k \in \N\), the problem 
	\(\MaxSAT{k}\) is the CSP on the set of
	variables \(\{x_1,\dots,x_n\}\), where 
	the constraint family \(\mathcal{C}\)
	is the set of all sat clauses which consist
	of at most \(k\) literals.
\end{definition}

\begin{corollary}
  \label{cor:MaxSAT}
	For infinitely many \(n\),
  \(\fc(\MaxSAT{2}, 1 - \varepsilon, 3/4 + \varepsilon) \ge
	n^{\Omega(\frac{\log n}{\log\log n})}\) and
  \(\fc(\MaxSAT{3}, 1 - \varepsilon, 3/4 + \varepsilon) \ge
	n^{\Omega(\frac{\log n}{\log\log n})}\)
  achieving inapproximability
	factor \(\frac{3}{4} + \Theta(\varepsilon)\),
        where \(n\) is the number of variables. In the case of SDPs,
        assuming Conjecture~\ref{conj:SDP-base}, we have
  \(\fcp(\MaxSAT{2}, 1 - \varepsilon, (1 + c_{GW}) / 2)\) and
  \(\fcp(\MaxSAT{3}, 1 - \varepsilon, (1 + c_{GW}) / 2)\)
        are both superpolynomial
        establishing inapproximability factor
        \(\frac{1+c_{GW}}{2} + \Theta(\varepsilon)
        \approx 0.93928 + \Theta(\varepsilon)\).
\begin{proof}
  We reduce \(\MaxCUT\) to \(\MaxSAT{2}\).
  For a 2-XOR
  clause \(l = (x_i \oplus x_j = 1)\) with \(i,j\in [n]\),
  we define two auxiliary constraints
  \(C_1(l) = (x_i \vee x_j)\) and \(C_2(l)
  = (\bar{x}_i \vee \bar{x}_j)\).
  Let \(\beta(L) \coloneqq \set{C_{1}(l), C_{2}(l)}{l \in L}\)
  for a set of 2-XOR clauses \(L\).
  We
  choose \(\gamma\) to be the identity map.
  Observe that whenever
  \(l\) is satisfied by a partial assignment \(s\)
  then both \(C_{1}(l)\)
  and \(C_{2}(l)\) are also satisfied by \(s\),
  otherwise exactly one of
  \(C_{1}(l)\) and \(C_{2}(l)\) is satisfied.
  Hence
  we obtain a reduction from \(\MaxCUT\) to \MaxSAT{2}.
  Together with
  Theorem~\ref{thm:HalfXOR} and Proposition~\ref{prop:probRed} the
  result follows. The statement for \(\MaxSAT{3}\) follows, as
  \(\MaxSAT{2}\) is a subproblem of \(\MaxSAT{3}\). 
\end{proof}
\end{corollary}

\subsection*{\MaxDicut{} inapproximability}

Problem \MaxDicut asks for a maximum sized cut
in a directed graph \(G\),
i.e., partitioning the vertex set \(V(G)\)
into two parts \(V_{0}\) and \(V_{1}\),
such that the number of directed edges \((i, j) \in E(G)\)
going from \(V_{0}\) to \(V_{1}\),
i.e., \(i \in V_{0}\) and \(j \in V_{1}\)
are maximal.
We use a formulation similar to \MaxCUT.

\begin{definition}[Directed Cut]
	For a fixed \(n\in \N\), the problem 
	\(\MaxDicut\) is the CSP with constraint 
	family \(\mathcal{C}_{DICUT} =
	\{\lnot x_i \land x_j \mid i,j \in [n], i\neq j\}\).
\end{definition}

We obtain \((1/2 + \varepsilon)\)-inapproximability
via the standard reduction from undirected graphs, by replacing every
edge with two, namely, one edge in either direction.  The
inapproximability factor is tight as the LP in \cite[Page~84,
Eq.~(DI)]{trevisan98._parallelLP}, is \(\frac{1}{2}\)-approximate
for maximum weighted directed cut.  In the SDP case we obtain 
\(c_{GW} + \varepsilon\)-inapproximability assuming  Conjecture~\ref{conj:SDP-base}
and in \cite{feige1995approximating} it is shown that \(\MaxDicut\)
can be approximated with a small SDP within a factor of 0.859 leaving
a (conditional) gap of about 0.02.

\begin{corollary}
  \label{cor:MaxDICUT}
  For infinitely many \(n\), we have
  \(\fc(\MaxDicut, 1/2 - \varepsilon, 1/4 + \varepsilon ) \ge
  n^{\Omega(\frac{\log n}{\log\log n})}\)
  achieving inapproximability factor \(1/2 + \Theta(\varepsilon)\).
        Assuming
        Conjecture~\ref{conj:SDP-base},
  in the SDP case
  \(\fcp(\MaxDicut, 1/2 - \varepsilon,
  c_{GW}/2 + \varepsilon)\)
  is superpolynomial
  establishing inapproximability factor
  \(c_{GW} + \Theta(\varepsilon)\).
\begin{proof}
The proof is analogous to that of Corollaries~\ref{cor:MaxCSP}
and~\ref{cor:MaxSAT},
hence we point out only the differences
We reduce \(\MaxCUT\) to \(\MaxDicut\),
but now replace every clause \(l = (x_i \oplus x_j = 1)\)
with \(C_{1}(l) = \lnot x_i \land x_j\)
and \(C_{2}(l) = x_i\land \lnot x_j\).
Now observe that whenever \(x_i\oplus x_j = 1\) is satisfied by 
	a partial assignment \(s\) then 
	exactly one of \(\lnot x_i \land x_j\) and 
	\(x_i\land \lnot x_j\) is also satisfied by \(s\). 
	If on the other hand \(x_i\oplus x_j\) is not fulfilled
	by \(s\), then neither \(\lnot x_i \land x_j\) nor
	\(x_i \land \lnot x_j\) are fulfilled.
The remainder of the proof is the same as in Corollary~\ref{cor:MaxSAT}.
\end{proof}
\end{corollary}

\subsection{Minimum constraint satisfaction}
\label{sec:minim-constr-satisf}

In this section we examine minimum constraint satisfaction problems,
a variant of constraint satisfaction problems,
where the objective is not to maximize the number of
\emph{satisfied} constraints,
but to minimize the number of \emph{unsatisfied} constraints.
This is equivalent to maximizing the number of satisfied
constraints, however, the changed objective function
yields different approximation factors
due to the change in the magnitude of the optimum value; this is in
analogy to the algorithmic world.
We consider only \MinCNF{2} and \MinUnCUT{} from
\cite{agarwal2005log}, which are complete in their class in the
algorithmic hierarchy; our technique applies to many more problems in
\cite{papadimitriou1991optimization}. The problem \(\MinCNF{2}\) is of
particular interest here,
as it is considered to be the hardest minimum
CSP with nontrivial approximation guarantees (see
\cite{agarwal2005log}). 
We start with the general definition of minimum CSPs. 

\begin{definition}
  The \emph{minimum Constraint Satisfaction Problem}
  on variables \(x_{1}, \dotsc, x_{n}\)
  with constraint family \(\mathcal{C} = \{C_{1}, \dotsc, C_{m}\}\)
  is the minimization problem with
  \begin{enumerate}
  \item \textbf{feasible solutions}
    all 0/1 assignments to \(x_1,\dots,x_n\);
  \item \textbf{instances}
    all nonnegative weightings \(w_{1}, \dotsc, w_{m}\)
    of the constraints \(C_{1}, \dotsc, C_{m}\);
  \item \textbf{objective functions}
    weighted sum of negated constraints, i.e.
    \(\val_{w_{1}, \dotsc, w_{m}}(x_1,\dots,x_n) =
    \sum_i w_i [1 - C_i(x_{1},\dots,x_{n})]\).
  \end{enumerate}
  The goal is to minimize the objective function,
  i.e., the weight of unsatisfied constraints.
\end{definition}

As mentioned above, we consider two examples.
\begin{example}[Minimum CSPs]
  The problem \(\MinCNF{2}\)
  is the minimum CSP with constraint family
  consisting of all disjunction of two literals,
  as in \MaxSAT{2}.
  The problem \(\MinUnCUT\)
  is the minimum CSP with constraint family
  consisting of all equations \(x_{i} \oplus x_{j} = b\)
  with \(b \in \{0,1\}\),
  as in \MaxXOR{2}.
\end{example}

We are ready to prove LP inapproximability bounds for these problems. 
Instead of the reductions in \cite{chlebik2004approximation}, we use
direct, simpler reductions from \MaxCUT{} and here we provide reductions for general weights. Note that the current best
known algorithmic inapproximability
for \(\MinCNF{2}\) is \(8\sqrt{5} - 15 - \varepsilon
\approx 2.88854 - \varepsilon\) by
\cite{chlebik2004approximation}.
Assuming the Unique Games Conjecture,
\cite{chawla2006hardness} establish that \(\MinCNF{2}\) cannot be
approximated within any constant factor and our LP inapproximability
factor coincides with this one. The problem \(\MinUnCUT\) is known to
be SNP-hard (see \cite{papadimitriou1991optimization}) however the
authors are not aware of the strongest known factor. We refer the
reader to \cite{khanna2001approximability} for a classification of all
minimum CSPs.

\begin{theorem}
  \label{thm:minCSP}
  For every \(\varepsilon > 0\) and infinitely many \(n\),
  we have
  \(\fc(\MinCNF{2}, \varepsilon, 1/4 - \varepsilon)
  = n^{\Omega(\frac{\log n}{\log \log n})}\) and
  \(\fc(\MinUnCUT, \varepsilon, 1/2 - \varepsilon)
  = n^{\Omega(\frac{\log n}{\log \log n})}\)
  establishing inapproximability
  within any constant factor,
  where \(n\) is the number of variables.
  Assuming Conjecture~\ref{conj:SDP-base},
  even in the SDP case we have that
  \(\fcp(\MinCNF{2}, \varepsilon, (1 - c_{GW}) / 2 - \varepsilon)\)
  and \(\fcp(\MinUnCUT, \varepsilon, 1 - c_{GW} - \varepsilon)\)
  are superpolynomial
  showing inapproximability within any constant factor.
\begin{proof}
We reduce from \MaxCUT to \MinCNF{2}
similar to the previous reductions:
assignments are mapped to themselves,
i.e., \(\gamma\) is the identity.
Under \(\beta\)
every clause \(C_{\ell}\) is replaced with two disjunctive clauses
\(C_{\ell}(1)\) and \(C_{\ell}(2)\),
both inheriting the weight \(w_{\ell}\) of \(C_{\ell}\),
i.e.,
\(\val^{\MinCNF{2}}_{\beta(w_{1}, \dotsc, w_{m})}
(x_{1}, \dotsc, x_{n}) =
\sum_{\ell} w_{\ell} [(1 - C_\ell(1)) + (1 - C_\ell(2))]\).

For \(C_\ell = (x_{i} \oplus x_{j} = 1)\),
we let \(C_\ell(1) \coloneqq  x_{i} \lor \lnot x_{j}\)
and \(C_\ell(2) \coloneqq \lnot x_{i} \lor x_{j}\).
Note that if \(C_\ell\) is unsatisfied,
then both of \(C_\ell(1)\) and \(C_\ell(2)\) are satisfied,
and if \(C_\ell\) is satisfied,
then exactly one of \(C_\ell(1)\) and \(C_\ell(2)\) is satisfied.
Therefore
\(\val^{\MinCNF{2}}_{\beta(w_{1}, \dotsc, w_{m})}
= \sum_{\ell} w_{\ell}
- \val^{\MaxCUT}_{w_{1}, \dotsc, w_{m}} w_{\ell}\).
This provides the desired lower bound.

For \MinUnCUT the reduction is similar but simpler,
as we replace every clause \(C\) with itself.
\end{proof}
\end{theorem}

\section{From matrix approximation to problem approximations}
\label{sec:round-pert}

We will now explain how a nonnegative matrix with small nonnegative
rank (or semidefinite rank) that is close to a slack matrix of a
problem \(\oprob\) of interest can be rounded to an actual slack
matrix with a moderate increase in nonnegative rank (or semidefinite
rank) and error. This argument is implicitly contained in
\cite{Rothvoss11,BDP2013jour} for the linear and semidefinite case
respectively. In some sense we might want to think of this approach
as an interpolation between a slack matrix (which corresponds to
\(\oprob\)) and a close-by matrix of low nonnegative rank (or
semidefinite rank) that does not correspond to any optimization
problem. The result is a low nonnegative rank (or semidefinite rank) approximation of \(\oprob\)
with small error. 

We will need the following simple lemma. Recall
that the exterior algebra of a vector space \(V\)
is the \(\R\)-algebra generated by \(V\)
subject to the relations \(v^{2} = 0\)
for all \(v \in V\).
As is customary, the product in this algebra is denoted by \(\wedge\).
The subspace of homogeneous degree-\(k\) elements
(i.e., linear combination of elements of the form
\(v_{1} \wedge \dotsb \wedge v_{k}\)
with \(v_{1}, \dotsc, v_{k} \in V\))
is denoted by \(\bigwedge^{k} V\).
Recall that for \(k = \dim V\),
the space \(\bigwedge^{k} V\) is one dimensional 
and is generated by \(v_{1} \wedge \dotsb \wedge v_{k}\)
for any basis \(v_{1}, \dotsc, v_{k}\) of \(V\).

\begin{lemma}
  \label{lem:small-matrix-factorzation}
  Let \(M \in \R^{m \times n}\) be a real matrix of rank \(r\).
  Then there are column vectors \(a_{1}, \dotsc, a_{r} \in \R^{m}\)
  and row vectors
  \(b_{1}, \dotsc, b_{r} \in \R^{n}\)
  with \(M = \sum_{i \in [r]} a_{i} b_{i}\).
  Moreover, \(\maxnorm{a_{i}} \leq 1\) and
  \(\maxnorm{b_{i}} \leq \maxnorm{M}\)
  for all \(1 \leq i \leq r\).
\begin{proof}
Consider the \(r\) dimensional vector space \(V\) spanned
by all the rows \(M_{1}, \dotsc, M_{m}\) of \(M\) and 
identify the one dimensional exterior product
\(\bigwedge^{r} V\) with \(\R\).
Now choose \(r\) rows \(M_{i_{1}}, \dotsc, M_{i_{r}}\)
for which \(M_{i_{1}} \wedge \dotsb \wedge M_{i_{r}}\)
is the largest in absolute value in \(\R\).
As the \(M_{i}\) together span \(V\) it follows that the largest value is non-zero. Hence \(M_{i_{1}}, \dotsc, M_{i_{r}}\) form a basis of \(V\).
Therefore any row \(M_{k}\) can be uniquely written as a linear
combination of the basis elements:
\begin{equation}
  \label{eq:small-matrix-factorization}
  M_{k} = \sum_{j \in [r]} a_{k,j} M_{i_{j}}.
\end{equation}
Fixing \(j \in [r]\)
and taking exterior products with the \(M_{i_{l}}\) where \(l \neq j\)
and both side we obtain
\begin{equation*}
  M_{i_{1}} \wedge \dotsc \wedge M_{i_{j-1}} \wedge
  M_{k}
   \wedge M_{i_{j+1}} \wedge \dotsc \wedge M_{i_{r}}
  =  a_{k,j} \cdot M_{i_{1}} \wedge \dotsc \wedge M_{i_{r}},
\end{equation*}
using the vanishing property of the exterior product.
By maximality of \(M_{i_{1}} \wedge \dotsc \wedge M_{i_{r}}\),
it follows that \(\abs{a_{k,j}} \leq 1\).
We choose \(a_{k} \coloneqq
\left[\begin{smallmatrix}
  a_{k,1} \\
  \dots \\
  a_{k,r}
\end{smallmatrix}\right]
\), and thus we have \(\maxnorm{a_{k}} \leq 1\).  Moreover choose
\(b_{j} \coloneqq M_{i_{j}}\), so that \(\maxnorm{b_{j}} \leq
\maxnorm{M}\) holds.  Finally, Eq.~\eqref{eq:small-matrix-factorization} can
be rewritten to \(M = \sum_{j \in [r]} a_{j} b_{j}\), finishing the
proof.
\end{proof}
\end{lemma}

For a vector \(a\) we can decompose it into its positive and negative part so that \(a = a^+ - a^-\) with \(a^+ a^- = 0\).
Let \(\abs{a}\) denote the vector obtained from \(a\)
by replacing every entry with its absolute value.
Note that
\(a^{+}\), \(a^{-}\), and \(\abs{a}\) are nonnegative vectors
and \(\abs{a} = a^{+} + a^{-}\).
Furthermore their \(\ell_{\infty}\)-norm is at most \(\maxnorm{a}\).

\begin{theorem}
  \label{thm:round-LP}
  Let \(\oprob\)
  be an optimization problem with \((C, S)\)-approximate
  slack matrix \(M\)
  and let \(\widetilde{M}\) be a nonnegative matrix.
  Then for the adjusted guarantee \(C'\) for \(\oprob\)
  defined as
  \begin{align}
      C'(f) &\coloneqq C(f) +
      (\rank M + \rank \widetilde{M}) \maxnorm{\widetilde{M} - M}
      &&\text{if \(\oprob\) is a maximization problem, and} \\
      C'(f) &\coloneqq C(f) - (\rank M + \rank \widetilde{M})
      \maxnorm{\widetilde{M} - M}
      &&\text{if \(\oprob\) is a minimization problem,}
  \end{align}
  we have
  \begin{align}
      \fc(\oprob, C', S) &\leq \LPrk \widetilde{M}
      + 2 (\rank M + \rank \widetilde{M})
      &&\text{and} \\
      \fcp(\oprob, C', S) &\leq \SDPrk \widetilde{M}
      + 2 (\rank M + \rank \widetilde{M}).
  \end{align}
\begin{proof}
We prove the statement for maximization problems; the minimization case follows similarly. The proof is based on the vector identity
\begin{equation}
  \label{eq:rank-k-nonneg-factor}
  \sum_{i \in [k]} \abs{a_{i}} b
  - \sum_{i \in [k]} a_{i} b_{i}
  =
  \sum_{i \in [k]} a_{i}^{+}  (b - b_{i})
  + \sum_{i \in [k]}  a_{i}^{-} (b + b_{i})
  .
\end{equation}

In our setting, the \(a_{i}\), \(b_{i}\) with $i \in [k]$ 
will arise from the (not necessarily nonnegative) factorization
of \(\widetilde{M} - M\),
obtained by applying Lemma~\ref{lem:small-matrix-factorzation}, i.e., we have
\begin{equation}
  \label{eq:round-LP-factor}
  \widetilde{M} - M
  =
  \sum_{i \in [k]} a_{i} b_{i},
\end{equation}
where \(\maxnorm{a_{i}} \leq 1\) and
\(\maxnorm{b_{i}} \leq \maxnorm{M}\)
for \(i \in [k]\) with
\(k \leq \rank (\widetilde{M} - M)
\leq \rank M + \rank \widetilde{M}\).
Furthermore, define
\(b \coloneqq \maxnorm{\widetilde{M} - M} \mathbb{1}\)
to be the row vector with all entries equal to \(\maxnorm{\widetilde{M} - M} \mathbb{1}\).

Substituting these values into~\eqref{eq:rank-k-nonneg-factor},
 using Eq.~\eqref{eq:round-LP-factor}
we obtain after rearranging
\begin{equation}
  \label{eq:2}
  N \coloneqq \sum_{i \in [k]} \abs{a_{i}} b
  + M
  =
  \widetilde{M}
  +
  \sum_{i \in [k]} a_{i}^{+}  (b - b_{i})
  +
  \sum_{i \in [k]} a_{i}^{-} (b + b_{i})
  ,
\end{equation}
so that we can conclude
\begin{align}
  \label{eq:4}
  \LPrk
    N
  &
  \leq \LPrk \widetilde{M}
  + 2 k \leq
  \LPrk \widetilde{M}
  + 2 (\rank M + \rank \widetilde{M}),
\\  \intertext{and similarly}
  \label{eq:3}
  \SDPrk N
  &
  \leq \SDPrk \widetilde{M}
  + 2 k \leq \SDPrk \widetilde{M}
  + 2 (\rank M + \rank \widetilde{M}).
\end{align}
It remains to relate $N$ to the \((C', S)\)-approximate slack matrix
of \(\oprob\).
By definition, the entries of \(N\) are
\begin{equation*}
  N(f, s) =
  \underbrace{\sum_{i \in [k]} \abs{a_{i}(f)} \cdot \maxnorm{\widetilde{M} - M}
  + C(f)}_{\coloneqq f^{*} \leq C'(f)} - \val_{f}(s),
\end{equation*}
where \(a_{i}(f)\) is the \(f\)-entry of \(a_{i}\).
Furthermore,
as \(\maxnorm{a_{i}} \leq 1\) and
\(k \leq \rank M + \rank \widetilde{M}\),
we have \(f^{*} \leq C'(f)\).
Thus the \((C', S)\)-approximate slack matrix \(M'\) of \(\oprob\)
looks like
\begin{equation*}
  M'(f, s) = N(f, s) + (C'(f) - f^{*}),
\end{equation*}
and as \(f^{*} \leq C'(f)\),
we have \(\LPrk M' \leq \LPrk N\) and \(\SDPrk M' \leq \SDPrk N\),
establishing the claimed complexity bounds
due to \eqref{eq:4} in the LP case and \eqref{eq:3} in the SDP case.
\end{proof}
\end{theorem}

A possible application of Theorem~\ref{thm:round-LP} is to
\lq{}thin-out\rq{} a given factorization of a slack matrix to obtain
an approximation with low nonnegative rank. The idea is that if a
nonnegative matrix factorization contains a large number of factors
that contribute only very little to each of the entries, then we can
simply drop those factors, significantly reduce the nonnegative rank,
and obtain a very good approximation of the original optimization
problem. Theorem~\ref{thm:round-LP} is then used to turn the
approximation of the matrix into an approximation of the original
problem of interest.

Also, it is possible to obtain low rank approximations of
combinatorial problems sampling rank-1 factors proportional to their
\(\ell_1\)-weight as done in the context of information-theoretic
approximations in \cite{BJLP2013}. However, the obtained
approximations tend to be too weak to be of interest.

\section{Final Remarks}
\label{sec:final-remarks}

We conclude with the following remarks:
\begin{enumerate}[leftmargin=1.5em, itemsep=-\parsep]
\item Our reduction mechanism works for both LPs and SDPs alike.
  Unfortunately, there are few inapproximability results for SDPs
  (such as the one in \cite{LRS14}), so
  that our SDP inapproximability factors are conditional on
  Conjecture~\ref{conj:SDP-base}, which serves as a blackbox here. 
\item We expect that via an appropriate reduction one can establish LP
  inapproximability of the TSP, however
  the current reductions, e.g., in \cite[§2]{TSPinapprox2001} and
  \cite[§6]{TSPinapprox2013}, cannot be used as they
  translate feasible solutions depending on the objective functions.
\item Our reduction could also establish high
  formulation complexity for problems in \(P\) by reducing from the
  matching problem.
\end{enumerate}

\section*{Acknowledgements}
\label{sec:acknowledgements}

Research reported in this paper was partially supported
by NSF grant CMMI-1300144 and NSF CAREER grant CMMI-1452463. The authors would like to thank James Lee
for the helpful discussions regarding max-CSPs.
We are indebted to Siu On Chan
for some of the PCP inapproximability bounds as well as Santosh
Vempala for the helpful discussions.

\bibliographystyle{abbrvnat}
\bibliography{bibs}
\end{document}